\documentclass[onecolumn,amsmath,amssymb,nofootinbib,superscriptaddress]{revtex4-2}
\usepackage{wasysym}
\usepackage[utf8]{inputenc}
\usepackage[margin=1in]{geometry}
\usepackage{amsmath, amsthm, amssymb,amscd, mathrsfs, amsfonts, mathtools,tikz-cd,pgfplots}
\usepackage{appendix}
\usepackage{relsize}
\usepackage{mathtools}
\usepackage{dsfont}
\usepackage{float}

\usepackage{bbding}
\usepackage{quantikz}
\usepackage[linesnumbered,ruled,vlined]{algorithm2e}
\usepackage{algpseudocode}
\usepackage{physics}
\usepackage{graphicx}
\usepackage{textcomp}
\SetKwInput{KwInput}{Input}                % Set the Input
\SetKwInput{KwOutput}{Output}              % set the Output
\usepackage{afterpage}
\usetikzlibrary{shapes, backgrounds}
\usepackage{soul}
\usepackage{makecell}
\usepackage{tabularx}
\newcolumntype{Y}{>{\centering\arraybackslash}X}
\tikzset{
  qaoaBlock/.style={
    draw,
    rounded corners=2pt,
    minimum width=1.4cm,
    minimum height=0.8cm,
    align=center,
    font=\small,
  },
  problem/.style={qaoaBlock, fill=green!30},
  mixer/.style={qaoaBlock, fill=blue!30},
}
\usetikzlibrary{decorations.pathmorphing}
\usetikzlibrary{arrows.meta}
\usetikzlibrary{patterns}
\usepackage{chngcntr}
\usepackage[bookmarks=true,
bookmarksnumbered=true,
breaklinks=true,
pdfstartview=FitH,
hyperfigures=false,
plainpages=false,
naturalnames=true,
colorlinks=true,
linkcolor=blue,       % Color for internal links
citecolor=blue,       % Color for citations
urlcolor=blue,        % Color for URLs
pagebackref=true,
pdfpagelabels]{hyperref}
\usepackage{ORCIDinREVTeX}
\theoremstyle{definition}
\newtheorem{thm}{Theorem}[section]
\newtheorem{prop}[thm]{Proposition}
\newtheorem{lem}[thm]{Lemma}

\newtheorem{defn}[thm]{Definition}
\newtheorem{cor}[thm]{Corollary}
\newtheorem{rmk}[thm]{Remark}

\newtheorem{ex}[thm]{Example}

\usepackage[normalem]{ulem} 

\begin{document}
\title{Efficient  Estimation of Reduced QAOA Expressibility on Geodetic Graphs}% \sout{Graph Structure Reveals Quantum Dynamics in Symmetry-Reduced QAOA on Trees.}}%: Theoretical Insights and Practical Implications}
%\title{Symmetry Principles for Variational Quantum Algorithms}
% 100 words: This work identifies a structural principle governing the expressivity–trainability trade-off in variational quantum algorithms. By linking instance symmetries to the dynamical Lie algebra that determines reachable quantum dynamics, we provably demonstrate that barren plateaus can be systematically controlled through symmetry-resolved transformations that preserve optimal solutions. We prove that modest graph modifications can enforce maximal reachable dynamics in reduced spaces and can qualitatively change scaling of circuit complexity. This establishes symmetry resolution as a general design paradigm for variational algorithms, providing scalable and actionable criteria for their design principles and lays the foundation for fundamentally new hybrid quantum-classical algorithmic constructions.

\author{Bao Bach} 
\orcid{0000-0001-6210-7725}
\email{baobach@udel.edu}
\affiliation{Department of Computer and Information Sciences, University of Delaware, Newark, DE 19716, USA} 
% \affiliation{Department of Quantum Science and Engineering, University of Delaware, Newark, DE 19716, USA} 

\author{Boris Tsvelikhovskiy} 
\orcid{0000-0003-0798-7218}
%\email{borist@ucr.edu}
\affiliation{Department of Mathematics, University of California, Riverside, CA 92521, USA}

\author{Jose Falla}
\orcid{0000-0001-9918-2198}
%\email{jfalla@udel.edu}
\affiliation{Department of Physics and Astronomy, University of Delaware, Newark, DE 19716, USA} 

\author{Ilya Safro} 
\orcid{0000-0001-6284-7408}
%\email{isafro@udel.edu}
\affiliation{Department of Computer and Information Sciences, University of Delaware, Newark, DE 19716, USA} 
\affiliation{Department of Physics and Astronomy, University of Delaware, Newark, DE 19716, USA}

\begin{abstract}

Classically equivalent formulations of an optimization problem need not lead to equivalent quantum algorithms. For MaxCut, fixing the value of a single vertex removes a simple global symmetry without changing the underlying optimization problem, yet it can substantially alter the quantum dynamics generated by the Quantum Approximate Optimization Algorithm (QAOA). These dynamics are captured by the circuit's dynamical Lie algebra (DLA), whose direct construction can become exponentially expensive. Here we show that, for symmetry-reduced QAOA on geodetic graphs, substantial information about the reduced DLA can instead be obtained efficiently from the graph alone. We introduce a polynomial classical algorithm that recursively distinguishes vertices using shortest path structure and degree parity. When all vertices are resolved individually, the method determines the complete reduced DLA and, under the corresponding graph conditions, certifies maximal expressibility of the reduced QAOA ansatz. Even when full resolution is not achieved, the algorithm identifies embedded subalgebras, provides rigorous lower bounds on DLA dimension, and certifies controllable subsystems. Our theoretical and experimental results show how classical graph structure can be used to diagnose and potentially guide the quantum dynamics of variational algorithms before running them on quantum hardware.

\end{abstract}

\maketitle

\hspace{-0.17in}\textbf{Keywords:} Quantum Approximate Optimization Algorithm, Dynamical Lie Algebras, Combinatorial Optimization, MaxCut.

\section{Introduction}

How can we determine what a quantum circuit is actually capable of doing? For the Quantum Approximate Optimization Algorithm (QAOA), the circuit itself is remarkably simple: it alternates evolution under only two types of Hamiltonians, one encoding the optimization problem and one mixing candidate solutions. Yet these two generators can collectively produce very different quantum dynamics depending on the underlying problem structure. They may restrict the evolution to a relatively small part of the Hilbert space, or they may generate dynamics rich enough to reach essentially every state available to the system. One of the recent mathematical abstractions that has attracted broad attention in quantum information and captures this distinction is the \emph{dynamical Lie algebra} (DLA).

Let us focus first on the meaning of a DLA for quantum circuits. A parameterized variational ansatz, such as the QAOA, is constructed from a finite set of Hamiltonians that determine the instantaneous directions of state evolution. By applying the exponentials of these Hamiltonians repeatedly, in different orders and for different evolution times, one generates increasingly complicated transformations. The DLA precisely collects all skew-Hermitian operators generated by these basic building blocks, which in turn dictate the reachable transformations. It therefore tells us what portion of the quantum dynamics is fundamentally accessible to the circuit, independently of whether a particular classical optimizer succeeds in finding the right parameters. A sufficiently rich DLA can certify the ability to prepare arbitrary states in the relevant Hilbert space, whereas symmetries and other structural features of the problem instance can confine the dynamics to smaller invariant subspaces. At the same time, the dimension and structure of the DLA are closely related to the geometry of the variational loss landscape and to the possible emergence of barren plateaus~\cite{FHCKYHSP,LJGCC,LCSMCC,RBSKMLC}. Thus, the DLA provides information about two central questions in variational quantum computing: \emph{What can the circuit express?} and \emph{How difficult is it to train?}

Determining this algebra directly, however, can be extremely expensive. Even when a quantum circuit is generated by only a few simple Hamiltonians, repeated commutation can produce a DLA whose dimension grows exponentially with the number of qubits. \emph{We show that, for an important class of symmetry-reduced QAOA instances, substantial (and in favorable cases complete) information about this exponentially large quantum dynamical object can instead be obtained by a polynomial-time classical algorithm operating only on the underlying graph.} 
For geodetic graphs (undirected graphs such that there exists a unique (unweighted) shortest path between every two vertices) $\Gamma=(V,E)$ with a chosen reference vertex $v\in V$, our algorithm uses shortest-path distances to $v$ together with vertex-degree parities profile along unique paths to partition $V\setminus\{v\}$, yielding explicit subset-supported operator sums that belong to the reduced DLA. Its overall complexity is  $\mathcal{O}\bigl(|V|\bigl(\abs{V} + \abs{E}\bigr)\bigr)$. When the resulting partition fully separates the vertices, the algorithm determines the entire reduced DLA and, under the corresponding graph-theoretic conditions, certifies maximal expressivity of the reduced QAOA ansatz. When complete separation is not achieved, the outcome still identifies large subalgebras within the DLA, rigorous lower bounds on the DLA dimension, and controllable quantum subsystems. 

In a geodetic graph, each vertex has a unique shortest path from the fixed vertex \(v\), so its degree-parity profile is well-defined and unique. One important subclass of geodetic graphs is acyclic graph/tree. Trees provide a tractable setting in which classical graph structure can be efficiently connected to reduced quantum dynamics (especially, it is trivial to solve MaxCut for this). In this case, the Lie-algebraic analysis only needs to look at the shortest path from \(v\) to the considered vertices. However, for other geodesic graphs that are not acyclic, two vertices at the same distance from \(v\) can also be adjacent. An edge \((a,b)\) of this kind contributes \(Z_aZ_b\) to the reduced cost Hamiltonian even though it lies on neither vertex’s shortest path from \(v\), which complicates the Lie-algebraic analysis (especially, it is not clear if this class of graphs can be solved efficiently for MaxCut). Thus, the tree case provides a natural starting point for the numeric analysis, along with more advanced cases with generated non-acyclic geodetic graphs.\\

\noindent {\bf QAOA: Expressivity vs Trainability~} 
This work is motivated by the broader challenge of understanding and improving QAOA~\cite{QAOA}, one of the most widely studied variational quantum algorithms for combinatorial optimization. In QAOA, a classical objective function is encoded into an Ising-type Hamiltonian whose low-energy states correspond to high-quality solutions of the optimization problem. The approach has been investigated for problems including MaxCut~\cite{QAOA}, community detection~\cite{shaydulin2019network}, graph coloring~\cite{graph_coloring}, and graph partitioning~\cite{ushijima2021multilevel}, among many others.

At its core, QAOA alternates between a \emph{problem Hamiltonian} $H_P$, which encodes the objective function, and a complementary \emph{mixer Hamiltonian} $H_M$, which drives transitions between candidate solutions. At circuit depth $p$, these alternating evolutions are controlled by two $p$-tuples of real variational parameters,
$\boldsymbol{\gamma}$ and $\boldsymbol{\beta}$, which are optimized classically. The quantum computer therefore explores a parameterized family of states, while the classical optimizer searches within this family for parameters producing a small expectation value of $H_P$.

This hybrid structure is both the strength of QAOA and one of its primary limitations. Increasing the circuit depth can enlarge the family of accessible states, but it simultaneously makes the classical parameter optimization more difficult. The resulting optimization landscapes are typically nonconvex, and the cost of finding good parameters can substantially reduce the potential quantum advantage~\cite{ZWCHL}. A particularly severe manifestation of this problem is the appearance of \emph{barren plateaus}, where gradients of the objective function become exponentially small with increasing system size, making conventional gradient-based optimization ineffective~\cite{mcclean2018barren,LTWS}. Consequently, one of the central challenges in QAOA is not merely to construct a highly expressive circuit, but to obtain a useful balance between \emph{expressivity} and \emph{trainability}.

A considerable body of research therefore focuses on improving the optimization of QAOA parameters. One direction is better parameter initialization. For example, the Beinit framework~\cite{beta} initializes variational parameters using a data-dependent beta distribution whose hyperparameters are inferred from the input instance and introduces controlled perturbations during gradient descent to reduce the probability of becoming trapped in flat regions of the loss landscape. Another important direction is parameter transferability. Optimal QAOA parameters have been observed to exhibit regularities associated with local graph structure, making it possible to reuse parameters obtained for smaller or structurally related instances~\cite{galda2023similarity,brandao2018fixed,falla2024graph,xu2025qaoa} or even across different combinatorial problems \cite{nguyen2025cross,nguyen2026graph,montanez2025transfer}. Such approaches can substantially reduce the classical optimization burden and provide useful initializations for larger problems. Other proposals include reinforcement-learning-based parameter initialization for deep circuits~\cite{peng2025breaking} and engineered Markovian dissipation between variational layers to improve trainability~\cite{sannia2024engineered}.

A complementary line of research modifies the quantum circuit itself. In particular, the mixer determines which transitions the circuit can generate between candidate solutions and can therefore impose important dynamical restrictions. Although the standard Pauli-$X$ mixer remains the most commonly studied architecture, alternative mixers can produce qualitatively different behavior. For instance, a complete DLA characterization has been recently obtained for Grover-mixer QAOA~\cite{TNB}. Constrained XY mixers, which are designed to preserve feasibility conditions during the quantum evolution, have also been analyzed from a Lie-algebraic perspective for connectivity structures such as all-to-all and ring topologies~\cite{kordonowy2025lie,HWORVB}. These results illustrate an important general principle: even when the optimization problem itself is unchanged, modifying the generators of the quantum circuit can fundamentally change the available quantum dynamics. Both parameter transferability and mixer design directions have been combined in the adaptive \cite{liu2022layer,ZTBCEBM} and more recently AI-generative systems \cite{tyagin2025qaoa,ugale2026q3sat}.\\

\noindent {\bf QAOA: The role of DLAs and symmetries~}
Dynamical Lie algebras provide a natural framework for putting these different observations on common mathematical ground. Formally, if a parameterized circuit is generated by Hamiltonians $H_1,\dots,H_m$, its DLA is the real Lie algebra generated by the skew-Hermitian operators $iH_1,\dots,iH_m$ under repeated commutation. For standard QAOA, the corresponding DLA is
\[
\mathfrak{g}_{\mathrm{std}}
=
\langle iH_M,iH_P\rangle_{\mathrm{Lie}}.
\]
The structure and dimension of this algebra characterize the transformations that can be generated by the ansatz and are therefore directly connected to controllability and expressive power~\cite{KLFCCZ,RBSKMLC}. DLAs also provide information about trainability and gradient concentration~\cite{FHCKYHSP,LJGCC,LCSMCC,RBSKMLC}. Their relevance extends beyond QAOA to quantum machine learning, where they provide rigorous criteria for expressivity and trainability and help identify regimes in which barren plateaus are unavoidable~\cite{GLCCS,LTWS,WHSU}. Lie-algebraic structure can additionally reveal classes of variational circuits that admit efficient classical descriptions~\cite{CLG,GLCCS}.

Despite their importance, DLAs associated with QAOA remain difficult to characterize in general. Existing analyses for QAOA with the standard Pauli-$X$ mixer have largely focused on particular graph families, including paths, cycles, and complete graphs~\cite{ASYZ1,KLFCCZ}. More recently, for MaxCut on Erd\H{o}s--R\'enyi random graphs $G(n,p)$ with $p=0.5$, the standard QAOA DLA was shown with high probability to have exponentially large dimension, with a structure consisting of two simple Lie-algebra components of dimension $\Theta(4^n)$~\cite{MYAZ}. The fact that such large algebras arise from only two structured QAOA Hamiltonians illustrates why explicit DLA construction becomes rapidly impractical.

An additional way to change the quantum dynamics arises even before the quantum circuit is constructed: one may change the classical representation of the optimization problem. Many binary optimization problems possess symmetries that identify equivalent solutions. MaxCut provides a canonical example. If a binary string $x$ specifies the two sides of a cut, simultaneously flipping every bit produces its complement $\bar{x}$, which represents exactly the same cut with its two parts exchanged. One may therefore fix the value of one binary variable, equivalently, assign one graph vertex to a prescribed side of the cut, without changing the underlying optimization problem. The use of symmetries to reduce QAOA instances has been considered in several previous works~\cite{TSA1,shaydulin2021classical,zhao2025symmetry}.

What appears trivial classically, however, can be highly nontrivial quantum mechanically. In~\cite{TBFS}, we showed that fixing a single symmetry-related variable modifies the Hamiltonian generators of QAOA and can substantially change both the structure and the dimension of the resulting \emph{reduced dynamical Lie algebra}. Different choices of the fixed vertex therefore give classically equivalent optimization problems but potentially very different parameterized quantum circuits. A particularly striking example is provided by $k$-armed spider graphs. For the unreduced QAOA instance, the dimension of the standard DLA grows exponentially with the number of vertices. Fixing the central vertex decomposes the reduced graph into independent paths, and the corresponding reduced DLA grows only quadratically~\cite{TBFS}. Thus, fixing a single redundant classical variable can change the algebraic complexity of the quantum dynamics from exponential to polynomial.

This observation raises the central algorithmic question considered in our paper: \emph{Given a particular symmetry reduction, can we efficiently determine what quantum dynamics it generates without explicitly recovering the reduced DLA?}\\

\noindent \textbf{Our contribution~} We answer this question for symmetry-reduced QAOA applied to MaxCut on geodetic graphs. We adapt the structural framework developed in~\cite{TBFS} into a polynomial-time classical algorithm that certifies elements of the standard reduced DLA associated with fixing a vertex $v\in V$. Given a graph $\Gamma$ and a source vertex $v\in V$ whose associated bit is fixed, the algorithm groups the remaining vertices by distance and by the degree parities along their unique shortest paths from the source. Whenever a vertex is isolated as a singleton, it can serve as a new source for further refinement.

The Lie-algebraic significance of this partition is the following. For each class $S$ produced by the procedure, the structural results of~\cite{TBFS} certify that a collective Pauli operator of the form
\[
i\sum_{w\in S}X_w
\]
belongs to the standard reduced DLA
$\mathfrak{g}^{\,v}_{\Gamma,\mathrm{std}}$.
Thus, if a class contains a single vertex $w$, the algorithm has therefore certified that the individual generator $iX_w$ belongs to the DLA. Such singleton vertices can then be used as new sources, producing additional partitions that are combined recursively to further separate the graph.

When this refinement process eventually separates every vertex in $V\setminus\{v\}$, all individual $X$ generators are recovered and we obtain
\[
\mathfrak{g}^{\,v}_{\Gamma,\mathrm{std}}
=
\mathfrak{g}^{\,v}_{\Gamma,\mathrm{free}}.
\]
This equality is particularly useful because free DLAs admit structural classifications~\cite{WKKB2,KLFCCZ}. 

\textbf{The novelty of the manuscript is not limited to the cases in which complete characterization of the DLA is achieved.} If the partition remains only partially refined, every singleton vertex still certifies the presence of an individual $iX_w$ generator. Let
\[
\mathcal{S}_{\Gamma,v}
=
\{w\in V\setminus\{v\}\mid
iX_w\in\mathfrak{g}^{\,v}_{\Gamma,\mathrm{std}}\}
\]
be the subset of such vertices. For subgraphs $\Gamma'$ supported on $\mathcal{S}_{\Gamma,v}$, we obtain the inclusion of Lie algebras
\[
\mathfrak{g}_{\Gamma',\mathrm{free}}
\subseteq
\mathfrak{g}^{\,v}_{\Gamma,\mathrm{std}}.
\]
This inclusion induces rigorous lower bounds on the dimension of the standard reduced DLA and, when $\Gamma'$ satisfies the corresponding controllability conditions, exact state-preparation guarantees on the associated subsystem. \emph{Thus, the output of the algorithm should not be viewed as a binary success or failure: partial refinement quantifies how much of the quantum dynamics can be rigorously certified.}

% Trees provide a particularly clean and useful setting for developing this framework. Between any two vertices there is a unique path, which means that the technical assumptions required by the structural constructions in~\cite{TBFS} are naturally satisfied without modifying the graph. The same property also makes the algorithm efficient. A breadth-first search on a tree from one source requires $\mathcal{O}(|V|)$ time, the common refinement of two partitions can be computed in $\mathcal{O}(|V|)$ time, and at most $\mathcal{O}(|V|)$ sources need to be processed. The resulting overall complexity of the proposed algorithm is therefore $\mathcal{O}\bigl(|V|^2\bigr)$.
Given the graph $G(V, E)$. For each source vertex, the breadth-first search computes the distance layers and unique predecessors in $\mathcal{O}(\abs{V}+\abs{E})$ time. The degree-parity profile of a vertex is obtained by extending its predecessor's profile by one bit. Storing the profiles in a binary tree therefore takes $\mathcal{O}(\abs{V})$ additional time. Each common refinement takes $\mathcal{O}(\abs{V})$ time, and at most $\abs{V}$ sources are processed. For a given fixed vertex $v$, the total running time is $\mathcal{O}\bigl(|V|\bigl(\abs{V} + \abs{E}\bigr)\bigr)$ for a general geodetic graph, which is reduced to $O(\abs{V}^2)$ on trees.

We further test how informative this classical procedure remains at system sizes for which direct DLA construction is completely impractical. For our first numerical experiments, we use $1,000$ random trees with up to $1,000$ vertices. The average size of the partition classes stabilizes near $1.27$, while the mean number of singleton classes exceeds $64\%$ of the vertices. For $1,000$-vertex graphs, the algorithm identifies on average approximately $643$ singleton vertices. Moreover, relative to the finest splitting permitted by graph automorphisms, the number of singletons found by the algorithm reaches approximately $96\%$ of the theoretical maximum. Similarly, we for non-acyclic geodetic graphs, we look at the infinite families of geodetic graphs created by joining two odd cycles, $F_{k}$ described in Section $6$ of~\cite{stober2023geodetic}. Looking at different sizes of $F_{k}$, we can see that the average number of the partition decreases to around $1.34$ as the graph size increases. Furthermore, even among different sizes, each non-acyclic geodetic graph $F_{k}$ has the number of singletons around $50\%$ of the number of vertices. For both benchmarks, although complete vertex separation becomes less common for large instances, the classical procedure continues to resolve a substantial fraction of the underlying quantum dynamics.

More broadly, these results suggest an expanded role for classical preprocessing in hybrid quantum algorithms. Classical computation is usually used to optimize the parameters of a quantum ansatz after the circuit architecture has already been chosen. Our results show that classical structural analysis can instead be used one level earlier: to determine what a candidate quantum representation is dynamically capable of doing. Because different symmetry-equivalent reductions of the same optimization problem may generate substantially different DLAs, an efficient estimate of their structure can provide information about controllability, expressivity, and potentially trainability before the quantum optimization itself is performed. This establishes a step toward symmetry-aware design of variational quantum algorithms in which classical graph structure is used to design a quantum circuit and also select and characterize the quantum dynamics that the circuit will generate. Figure~\ref{fig:placeholder} summarizes the framework and its main consequences. Our code is available at  \url{https://github.com/joseluisfalla/QAOA_Reductions_by_Classical_Symmetries}.

\begin{figure}
    \centering
     \includegraphics[width=\linewidth]{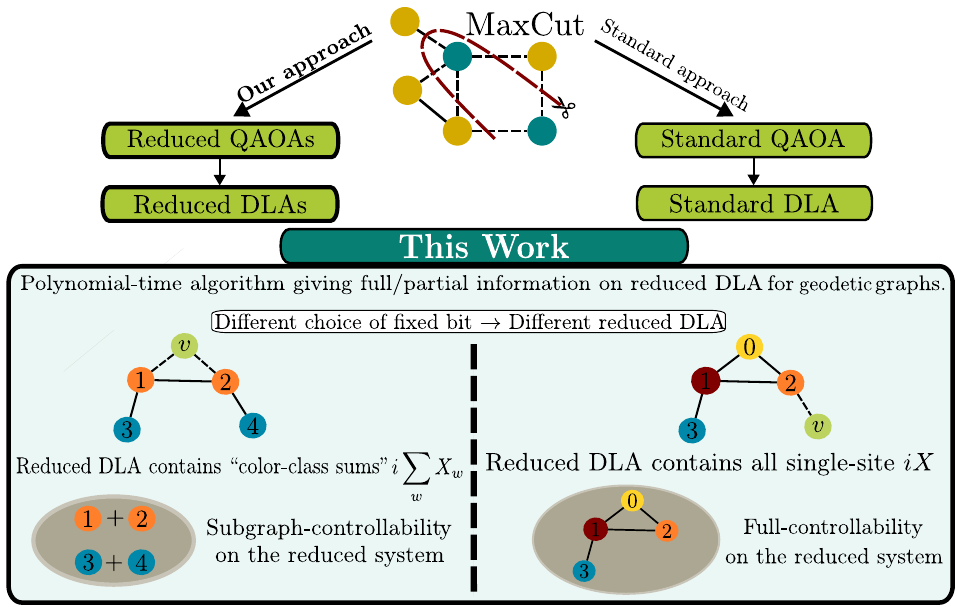}
\caption{\textbf{Overview of results and symmetry-reduction framework.} 
 In the presence of a global bit-flip symmetry, fixing the bit value at any chosen reference vertex $v$ yields a reduced optimization problem equivalent to the original, along with a reduced QAOA ansatz and its corresponding reduced dynamical Lie algebra. A canonical example exhibiting this symmetry is the MaxCut problem, which serves as the primary focus of this work.
\textbf{(Bottom)} We present a polynomial-time classical algorithm with computational complexity $\mathcal{O}\bigl(|V|\bigl(\abs{V} + \abs{E}\bigr)\bigr)$ that partitions the vertices of connected geodetic graphs based on their shortest-path distances to a chosen reference vertex $v$ and the vertex degree parities along those paths. Selecting different reference vertices $v$ induces distinct partitions, leading to two operational regimes: 
\textbf{(Left) Partial refinement:} certifies subalgebra inclusions, establishing rigorous lower bounds on the reduced DLA dimension and guaranteeing conditional subsystem reachability. Specifically, when using the algorithm, for the given example, when fixing node $0$, we are guaranteed that $iX_{1} + iX_{2} \in \mathfrak{g}^{v = 0}_{\Gamma, \mathrm{std}}$ and $iX_{3} + iX_{4} \in \mathfrak{g}^{v = 0}_{\Gamma, \mathrm{std}}$.
\textbf{(Right) Full refinement:} identifies the reduced
DLA; pure-state controllability follows under the conditions of Section \ref{sec: implications}. Specifically, when using the algorithm, for the given example, when fixing node $4$, we are guaranteed that $ iX_{0}, iX_{1}, iX_{2}, iX_{3}\in \mathfrak{g}^{v = 4}_{\Gamma, \mathrm{std}}$. Consequently, this gives us full controllability when we take the Lie closure with the cost Hamiltonian term.}
    \label{fig:placeholder}
\end{figure}

\section{Background}

%Variational quantum algorithms (VQAs) constitute a class of hybrid quantum--classical methods in which a parameterized quantum circuit is optimized using classical routines to minimize a cost function, typically expressed as the expectation value of a Hamiltonian~\cite{CABB}. A prominent instance of this paradigm is the Quantum Approximate Optimization Algorithm (QAOA), designed for solving combinatorial optimization problems~\cite{QAOA}. 

\textbf{Notations.} Let $\mathbb{B}^n = \{0,1\}^n$ denote the set of binary strings of length $n$, and let $W = (\mathbb{C}^2)^{\otimes n}$ be the corresponding Hilbert space with computational basis $\{\ket{x}\}_{x \in \mathbb{B}^n}$. Let $\Gamma = (V, E)$ be an undirected graph with vertex set $V = \{1, \dots, n\}$

\textbf{Quantum Approximate Optimization Algorithm.} In QAOA, a classical objective function defined over binary variables is encoded into a quantum operator, referred to as the \emph{problem Hamiltonian}, whose ground state corresponds to an optimal solution. 

The QAOA algorithm alternates unitary evolutions generated by a \emph{problem Hamiltonian} $H_P$ and a \emph{mixer Hamiltonian} $H_M$, which drives transitions between computational basis states to explore the solution space.
Starting from an initial state $\ket{\xi}$, typically chosen as the ground state of $-H_M$, the QAOA ansatz at depth $p$ is given by the unitary operator
\begin{equation}
\label{qaoa-chain}
U(\boldsymbol{\beta},\boldsymbol{\gamma})
  := e^{-i \beta_p H_M} e^{-i \gamma_p H_P} \cdots e^{-i \beta_1 H_M} e^{-i \gamma_1 H_P},
\end{equation}
where $\boldsymbol{\beta} = (\beta_1,\dots,\beta_p)$ and $\boldsymbol{\gamma} = (\gamma_1,\dots,\gamma_p)$ are 
arrays of real variational parameters. The parameterized unitary $U(\boldsymbol{\beta},\boldsymbol{\gamma})$ prepares the variational state
$\ket{\psi(\boldsymbol{\beta},\boldsymbol{\gamma})} = U(\boldsymbol{\beta},\boldsymbol{\gamma}) \ket{\xi}$. Here, the parameters are optimized by evaluating the objective through the expectation value $   \bra{\psi(\boldsymbol{\beta},\boldsymbol{\gamma})} H_P \ket{\psi(\boldsymbol{\beta},\boldsymbol{\gamma})}$
which is used as the loss function in the classical optimization loop. A schematic representation of this procedure is shown in Fig.~\ref{fig:qaoa_scheme}. 
Further details can be found in~\cite{BBCC, QAOA, HWORVB}.

\begin{figure}[htbp]
\includegraphics[width=0.8\textwidth]{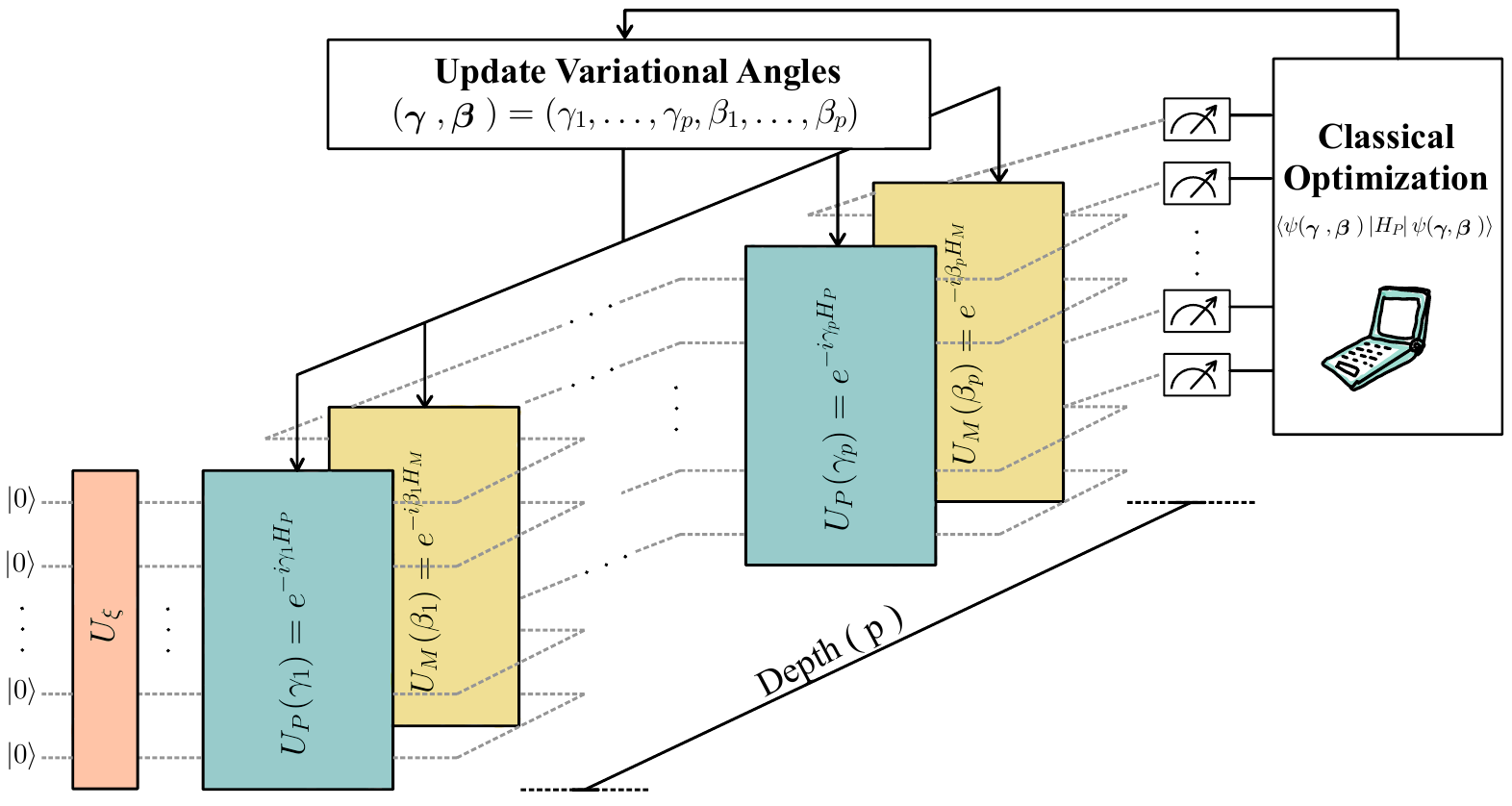}
\caption{\label{fig:qaoa_scheme} \textbf{Schematic illustration of the QAOA circuit.} The initial state preparation unitary $U_\xi$ with $U_{\xi}(\ket{0}^{\otimes n}) = \ket{\xi}$) is followed by $p$ alternating unitaries $U_P(\gamma_j):=e^{-i \gamma_j H_P}$ and $U_M(\beta_j):=e^{-i \beta_j H_M}$. At the end of the circuit, we measure the state in the computational basis. Each measurement outcome $x \in \mathbb{B}^n$ is assigned the value $F(x)$ of the objective function, 
and the empirical mean of these values provides an estimate of 
$\bra{\psi(\boldsymbol{\beta},\boldsymbol{\gamma})} H_P \ket{\psi(\boldsymbol{\beta},\boldsymbol{\gamma})}$. 
This estimate is then used in a classical optimization loop to update the parameters $(\boldsymbol{\beta}, \boldsymbol{\gamma})$ with the goal of minimizing the empirical mean.}
\end{figure}

\par{}
\textbf{Dynamical Lie Algebras.} Dynamical Lie algebras (DLAs) provide a powerful framework for analyzing the capabilities of variational quantum algorithms.

\begin{defn}
\label{defn:DLA}
Let $H_1, \dots, H_m$ be the Hermitian operators (Hamiltonians) that generate the unitary gates comprising a parameterized quantum circuit. The associated \textit{dynamical Lie algebra} is defined as the real Lie algebra generated by these operators under the standard Lie bracket: 
% \JF{This is sort of already (not formally) defined in section "QAOA: The role of DLAs and symmetries" of the Introduction. Does it seem redundant?}:
\begin{equation}
\mathfrak{g} := \langle iH_1, \dots, iH_m \rangle_{\mathrm{Lie}}.
\end{equation}

In the specific case of the Quantum Approximate Optimization Algorithm, the \emph{standard dynamical Lie algebra} is the real Lie algebra generated by the mixer and problem Hamiltonians, $H_M$ and $H_P$:
\begin{equation}\label{eq:stdDLA}
\mathfrak{g}_{\mathrm{std}} := \langle iH_M, iH_P \rangle_{\mathrm{Lie}}.
\end{equation}
\end{defn}

More generally, if the Hamiltonians admit decompositions into local terms,
\[
H_M = \sum_j M_j, \qquad H_P = \sum_k P_k,
\]
one also defines the \emph{free dynamical Lie algebra} as the Lie algebra generated by the individual local terms:
\begin{equation}
\mathfrak{g}_{\mathrm{free}} := \left\langle \{iM_j\}, \{iP_k\} \right\rangle_{\mathrm{Lie}}.
\end{equation}

The generated Lie algebra and its representation determine the associated connected dynamical group and therefore provide the standard algebraic description of arbitrary-depth reachability of the ansatz. Therefore, the analysis of DLAs is closely tied to the ansatz expressibility and trainability~\cite{KLFCCZ, RBSKMLC}. 
In particular, DLAs provide insight into the effective Hilbert space explored by QAOA and the geometry of the associated optimization landscape.

However, a fundamental challenge in the Lie-algebraic analysis of QAOA is the rapid growth of the associated DLAs with system size. For a considerable period, obtaining structural information about standard DLAs proved to be challenging, with rigorous results largely limited to specific graph families. A notable exception is the recent work of~\cite{MYAZ}, which provides new insight and suggests that, in generic settings, the standard DLA may be expected to coincide with the free one. Since free DLAs admit a complete classification (see \cite{WKKB2} for general result and see Theorem 1~\cite{KLFCCZ} for MaxCut),
%\BT{for MaxCut!!! The general result appears in \cite{WKKB2}}
in cases where
\begin{equation}
\label{eq:EqualityOfDLAs}
\mathfrak{g}_{\mathrm{std}}=\mathfrak{g}_{\mathrm{free}},
\end{equation}
one can immediately determine the precise structure of the standard DLA. This equality has previously been established only for the MaxCut problem on Erd\H{o}s--R\'enyi random graphs with edge creation probability $p=1/2$, where it holds with high probability $1-\exp(-\Omega(n))$ (see Theorem~3 in~\cite{MYAZ}).

\textbf{MaxCut Problem.} The MaxCut problem seeks a bipartition of vertices that maximizes the number of edges crossing between the two subsets. A partition can be encoded by a bit string $x \in \mathbb{B}^n$, where each $x_i \in \{0,1\}$ denotes the subset containing vertex $i$. An edge $(i, j) \in E$ is cut if and only if $x_i \neq x_j$. Accordingly, the objective function can be written as
\begin{equation}
    \sum_{(i,j)\in E} \left[(1 - x_i)x_j + x_i(1 - x_j)\right],
\end{equation}
which counts the total number of edges crossing the partition. The corresponding problem and standard mixer Hamiltonians for the MaxCut problem are
\begin{equation}
    H_M = \sum\limits_{w \in V} X_w, 
    \qquad
    H_P = \sum\limits_{(w,w') \in E} Z_w Z_{w'},
\end{equation}
which in turn give rise to the \emph{free} dynamical Lie algebra
\begin{equation}
    \mathfrak{g}_{\Gamma,\mathrm{free}}
= \Big\langle 
    \{ i X_w \mid w \in  V\}, \;
    \{ i Z_w Z_{w'} \mid (w w') \in E \} 
\Big\rangle_{\mathrm{Lie}},
\end{equation}
and the \emph{standard} dynamical Lie algebra
\begin{equation}
    \mathfrak{g}_{\Gamma,\mathrm{std}}
= \Big\langle  \mathcal{X}, \mathcal{Z}\Big\rangle_{\mathrm{Lie}}.
\end{equation}
with $\mathcal{X}:=iH_M$ and $\mathcal{Z}:=iH_P$.

In this work, the MaxCut problem will serve as the primary target problem for our Lie-algebraic analysis.
%\BT{It is important to set the scope!} \BB{We already said the scope is clear in the Introduction. The Background is optional for readers, so I don't think we should have an important statement here.}

\textbf{Symmetries and Reduced QAOAs.} Many combinatorial optimization problems, including the MaxCut problem, possess symmetries that can reduce the problem size.  In particular, the global bit-flip symmetry, under which the transformation $x \mapsto 1-x$ leaves the binary objective invariant, allows one variable to be fixed without loss of generality and reduces the problem to $n-1$ variables.

In the quantum setting, this classical reduction induces modified Hamiltonians acting on a reduced Hilbert space. This space is obtained by contracting the $2$-dimensional vector space, $\mathbb{C}^2$, associated with the fixed bit $v$, to the one-dimensional subspace $\mathbb{C}|c\rangle$ corresponding to the fixed classical value $c \in \{0, 1\}$. Thus the reduced Hilbert space is defined as
\begin{equation}
    \label{eq:reduced_hilbert}
    W_j = \mathbb{C}^2 \otimes \cdots \otimes \underbrace{\mathbb{C}|c\rangle}_{v\text{-th position}} \otimes \cdots \otimes \mathbb{C}^2 \simeq (\mathbb{C}^2)^{\otimes (n-1)}.
\end{equation}
Within this reduced space, operators acting on the fixed qubit are replaced by their corresponding scalar eigenvalues.

For MaxCut, fixing the bit value of one vertex converts interactions involving that vertex to single-qubit terms, yielding a reduced problem Hamiltonian. The resulting QAOA instance then generates a \emph{reduced dynamical Lie algebra}, defined analogously to the standard case but on the reduced system. By fixing the value of a single vertex, say $x_v = 0$, we thereby select a unique representative from each equivalence class. The resulting reduced problem is then defined on $n-1$ variables and has an objective function
\begin{equation}
 \sum_{\substack{(i,j)\in E \\ i,j \neq v}} \left[(1 - x_i)x_j + x_i(1 - x_j)\right]
+ \sum_{(v,j)\in E} x_j,
\end{equation}
where the second term accounts for edges incident to the fixed vertex. By fixing the bit value of vertex $v \in V$, the corresponding \emph{reduced} mixer and problem Hamiltonians become
\begin{equation}
\label{eq:reduced_hams}
\begin{aligned}
      H_M^v &= \sum\limits_{w \in V \setminus \{v\}} X_w, 
    \\
    H_P^v &= \sum\limits_{\substack{(w,w') \in E \\ w,w' \neq v}} Z_w Z_{w'} 
    + \sum\limits_{(v,w) \in E} Z_w,
\end{aligned}
\end{equation}
and the corresponding \emph{free} reduced  dynamical Lie algebra is defined as
\begin{equation}
\label{eq:reduced_free}
\mathfrak{g}^{\,v}_{\Gamma,\mathrm{free}}
= \Big\langle \{ i X_w \mid w \in V \setminus \{v\} \}, \{ i Z_w Z_{w'} \mid (w w') \in E, \, w,w' \neq v \}, \{ i Z_w \mid (v w) \in E \}
\Big\rangle_{\mathrm{Lie}},
\end{equation}
while the \emph{standard}  reduced Lie algebra is defined as
\begin{equation}
\label{eq:reduced_standard}
\mathfrak{g}^{\,v}_{\Gamma,\mathrm{std}}
= \Big\langle 
    \mathcal{X}_{\widehat{v}}, \; \mathcal{Z}_{\widehat{v}}
\Big\rangle_{\mathrm{Lie}},
\end{equation}
with generators
\begin{equation}
\mathcal{X}_{\widehat{v}} := iH_M^v 
\text{ and }
\mathcal{Z}_{\widehat{v}} := iH_P^v.    
\end{equation}

These reduced formulations highlight the interplay between problem structure and the associated Lie-algebraic dynamics. In particular, the dimension and structure of the reduced DLA may differ significantly from those of the original system, providing a useful setting for analyzing how graph structure and symmetry influence the behavior of QAOA (see~\cite{TBFS}).

    \begin{rmk}
        Following the definitions above, the inclusions below capture the relationship between the standard and free dynamical Lie algebras:
        \begin{equation}
            \begin{aligned}
                \mathfrak{g}_{\Gamma, \text{std}} \subseteq \mathfrak{g}_{\Gamma, \text{free}},\\
                \mathfrak{g}^{v}_{\Gamma, \text{std}} \subseteq \mathfrak{g}^{v}_{\Gamma, \text{free}}
        \end{aligned}
        \end{equation}
    \end{rmk}

\section{Methodology}
\label{sec:Method}

In this section, we leverage the properties of standard reduced DLAs established in~\cite{TBFS} to construct an algorithm that determines the containment of certain designated elements within $\mathfrak{g}^{v}_{\Gamma, \text{std}}$ for a pair $(\Gamma, v)$ consisting of an acyclic graph $\Gamma$ and a distinguished node $v$. We focus on geodetic graphs because this class naturally satisfies the structural conditions formulated in~\cite{TBFS}, which are essential to guarantee that elements reconstructed from the algorithm's output reside within the algebra. For complete technical details, we refer the reader to Appendix~\ref{sec:groundwork} and the references therein.

Given a connected geodetic graph $\Gamma=(V,E)$ and a vertex $v\in V$, our goal is to construct a partition of the remaining vertices, $V\setminus\{v\}$, for which each block \(S\) satisfies \(i\sum_{u\in S}X_u \in \mathfrak g_{\Gamma,\mathrm{std}}^v\)  (see Appendix~\ref{sec:groundwork}). We denote this partition by
\begin{equation}
    \label{eq:Partition}
    \text{\ScissorHollowRight}_{v} := \bigsqcup_{i=1}^{k_v} V_i,
\end{equation}
which is obtained by implementing the procedure based on the foundation in Section~\ref{sec:groundwork}. In general, the initial partition, ~$\underset{v, \text{init}}{\text{\ScissorHollowRight}}$, is first induced by the vertex subsets $C^{v}_{j,\mathbf{a}}$ (see Definition~\ref{defn:vertexSubsets}). We then recursively refine this partition using analogous subsets $C^{w}_{j,\mathbf{a}}$ for any vertex $w$ that emerges as a singleton during an intermediate step. The ultimate objective is to maximize the number of singleton subsets, $\# \{V_i \mid |V_i|=1\}$, in the final partition.

We search for this optimal configuration using the framework formalized in Algorithm~\ref{alg:recursive-partition}, which consists of two main subroutines: $\texttt{Partition}$ and $\texttt{Refinement}$. 

\subsection{Partition Subroutine}
Given the graph~$\Gamma=(V,E)$ and the source vertex~$v \in V$, we execute a Breadth-First Search originating from~$v$. This $\texttt{BFS}(v)$ traversal computes the shortest-path distance 
\begin{equation}
    \label{eq:Distance}
    D(u) := \operatorname{dist}_{\Gamma}(v,u)
\end{equation}
for every vertex~$u \in V$. Additionally, for each target vertex~$u \neq v$, we record the degree parity sequence~$\mathbf{a}$ associated with the path connecting~$v$ to~$u$. 
% These structural patterns are then organized and stored in a nested list~$A$, indexed either by the corresponding vertex itself or, equivalently, by the coordinate pair~$(D(u), \mathbf{a})$. \BT{I think this should be revised: we have exactly one sequence for every vertex (as we work with trees).}
Thus, this BFS procedure induces an initial partition

\begin{equation}
    \label{eq:InitialSplitting}
\underset{v, \text{init}}{\text{\ScissorHollowRight}}:=\bigsqcup\limits_{i} V_i,
\end{equation}

of the remaining vertices, where each equivalence class~$V_{i} = C^{v}_{j, \mathbf{a}}$ comprises vertices sharing both the same shortest-path distance from~$v$ and the same degree-parity pattern along this path (see Definition~\ref{defn:vertexSubsets}). We denote this entire classification subroutine as $\texttt{Partition}(v)$. It is well known that breadth-first search can be implemented in time $\mathcal{O}(|V|+|E|)$ (see, e.g., Chapter~20 of~\cite{cormen2022introduction}). For a connected acyclic graph, $|E| = |V|-1$, simplifying this time complexity to $\mathcal{O}(|V|)$.

\subsection{Refinement Subroutine}
Based on the initial partition~$\underset{v, \text{init}}{\text{\ScissorHollowRight}}$, we select a singleton vertex~$w \in \{h \in V_i \mid |V_i|=1\}$, provided one exists, and execute the subroutine $\texttt{Partition}(w)$ to obtain the resulting partition~$\text{\ScissorHollowRight}_{w}$. Next, we apply the subroutine $\texttt{Refinement}\left(\underset{v, \text{init}}{\text{\ScissorHollowRight}}, \text{\ScissorHollowRight}_{w}\right)$, which refines the initial partition based on this latest singleton grouping.  %In details, for each vertex $u \in V \setminus\{v\}$, find $\{V^{w}_{i}\} \in \text{\ScissorHollowRight}_{w}$ and $\{V^{v}_{i}\} \in \underset{v, \text{init}}{\text{\ScissorHollowRight}}$ such that $u \in \{V^{w}_{i}\}, \{V^{v}_{i}\}$. If $\{V^{w}_{i}\} \not\equiv \{V^{v}_{i}\}$, we perform $\{V^{v}_{i}\} \setminus \{V^{w}_{i}\}$ and $\{V^{v}_{i}\} \cap \{V^{w}_{i}\}$ to obtain new possible groupings. 
This refined partition is then designated as the current partition; that is,

\begin{equation}
    \label{eq:RefinedSplitting}
\underset{v, \text{current}}{\text{\ScissorHollowRight}}:=\texttt{Refinement}\left(\underset{v, \text{init}}{\text{\ScissorHollowRight}},\text{\ScissorHollowRight}_{w}\right)
\end{equation}

In more details, given two partitions $\text{\ScissorHollowRight}_1$ and $\text{\ScissorHollowRight}_2$ of the vertex set $V \setminus \{v\}$, the \textit{coarsest common refinement} formed by the pairwise intersections of the subsets from $\text{\ScissorHollowRight}_1$ and $\text{\ScissorHollowRight}_2$ can be computed efficiently by reducing it to the lexicographic sorting problem detailed in Section~2 of~\cite{PT}.

To formally establish this equivalence, let the blocks within each partition be indexed by
$$\text{\ScissorHollowRight}_1 = \{A_1, A_2, \dots, A_{|\text{\ScissorHollowRight}_1|}\}$$ and $$\text{\ScissorHollowRight}_2 = \{B_1, B_2, \dots, B_{|\text{\ScissorHollowRight}_2|}\}.$$ We can construct a map $\varphi: V\setminus\{v\} \to \Sigma^2$ that assigns to each vertex $u \in V\setminus\{v\}$ a string of length $2$ over the alphabet $\Sigma = \{1, 2, \dots, n\}$. The string representation is defined as:
\begin{equation}
\label{eq:CommonPartitionAndLexicogrSortingEquiv}
\varphi(u) = (\sigma_1(u), \sigma_2(u))
\end{equation}
where $\sigma_1(u) = i$ if $u \in A_i$, and $\sigma_2(u) = j$ if $u \in B_j$. Because every block in a partition is non-empty, the total number of blocks satisfies $\abs{\text{\ScissorHollowRight}_1} \le n-1$ and $\left|\text{\ScissorHollowRight}_2\right| \le n-1$. Thus, the maximum alphabet size $k = \max\left(\left|\text{\ScissorHollowRight}_1\right|, \left|\text{\ScissorHollowRight}_2\right|\right)$ is bounded by $n$.

Under this mapping, two vertices $u, w \in V\setminus\{v\}$ belong to the same block in the coarsest common refinement if and only if they share identical block memberships across both partitions, which means $u \in A_i \cap B_j$ and $w \in A_i \cap B_j$. This condition holds if and only if $\sigma_1(u) = \sigma_1(w)$ and $\sigma_2(u) = \sigma_2(w)$, meaning their assigned strings are identical ($\varphi(u) = \varphi(w)$). 

Once the assigned strings are determined, the grouping of vertices can be implemented either by inserting the pairs $(\sigma_1(u), \sigma_2(u))$ into a dictionary, or deterministically via radix sorting, since the maximum label value is bounded by $\abs{V}$. In either approach, the refinement step takes  $\mathcal{O}(\abs{V})$ time. Consequently, each call to the \texttt{Refinement} subroutine operates in linear time and is asymptotically equivalent to the preceding breadth-first search computation of \texttt{Partition}($w$), which requires $\mathcal{O}(\abs{V})$ time.

% In Section~2 of~\cite{PT} it was established that a multiset of $n$ strings can be lexicographically sorted in $\mathcal{O}(n + p)$ time, where $p$ is the total length of the distinguishing prefixes of the strings, defined as the smallest prefixes sufficient to distinguish the strings from each other. In our case, every string has length $2$. Consequently, the length of the distinguishing prefix for any string is at most $2$. Summing over all vertices, the total length of these distinguishing prefixes satisfies $p \le 2(|V|-1)$.% = \mathcal{O}(|V|)$.

% Once sorted, the elements of identical intersection blocks \BT{???} are arranged adjacently in the output array.\BT{???} A final, single linear pass over the sorted array compares adjacent pairs $\varphi(u_i)$ and $\varphi(u_{i+1})$ in $\mathcal{O}(1)$ time per step. A new partition block is generated whenever a component of the string tuple changes. This read-off phase preserves the optimal time bound, computing the explicit coarsest common refinement in $\mathcal{O}(|V|)$ total time. Therefore, this subroutine can be cast as a lexicographic sorting problem, allowing it to be solved in $\mathcal{O}(\abs{V} + \abs{E})$ time. \BT{I am not sure what is going on in this paragraph. Please revise it carefully.}

The \texttt{Partition} and \texttt{Refinement} subroutines are then repeated until either there are no remaining singleton vertices to process, or none of the refinement steps result in a strictly finer partition. Consequently, the algorithm executes at most $\mathcal{O}(\abs{V})$ iterations. Since each iteration requires $\mathcal{O}(\abs{V})$ time, as established previously, the overall running time is the product of the maximum number of iterations and the cost per iteration. The following theorem summarizes the structural guarantees and computational complexity of Algorithm~\ref{alg:recursive-partition}. The proof of the Theorem \ref{thm:algorithm_certificate}  can be found in Appendix \ref{sec:groundwork}

\begin{thm}
    \label{thm:algorithm_certificate}
Let $\Gamma = (V, E)$ be a connected geodetic graph and let $v \in V$ be a distinguished vertex. Algorithm~\ref{alg:recursive-partition} returns a partition $\text{\ScissorHollowRight}_{v}$ of $V \setminus \{v\}$ satisfying the following properties:
\begin{enumerate}
    \item For each block $V_i \in \text{\ScissorHollowRight}_{v}$,
    \begin{equation*}
        i \sum_{w \in V_i} X_{w} \in \mathfrak{g}^{v}_{\Gamma, \mathrm{std}}.
    \end{equation*}

    \item Every singleton $\{w\} \in \text{\ScissorHollowRight}_{v}$ guarantees $iX_{w} \in \mathfrak{g}^{v}_{\Gamma, \mathrm{std}}$.

    \item If every block $V_i \in \text{\ScissorHollowRight}_{v}$ is a singleton, then $\mathfrak{g}^{v}_{\Gamma, \mathrm{std}} = \mathfrak{g}^{v}_{\Gamma, \mathrm{free}}$.

    \item Algorithm~\ref{alg:recursive-partition} has a worst-case time complexity of $\mathcal{O}\bigl(|V|\bigl(\abs{V} + \abs{E}\bigr)\bigr)$ when $\Gamma$ is not acyclic and $\mathcal{O}\bigl(|V|^2\bigr)$ when $\Gamma$ is acyclic.
\end{enumerate}
\end{thm}

\subsection{Concrete Example}
We now demonstrate the execution of our program (Algorithm~\ref{alg:recursive-partition}) on a concrete example.

\begin{ex}
    Let $\Gamma$ be a tree graph with the set of vertices $V = \{1, 2, 3, 4, 5, 6, 7, 8\}$ and edges $E = \{(1, 2), (1, 3), (1, 4), (2, 5), (2, 6), (3, 7), (5, 8)\}$, as shown in the left panel of Figure~\ref{fig:step1_partition}. Let the fixed node be $v = 1$; consequently, our objective is to partition the remaining vertices of the tree, $V \setminus \{v\} = \{2, 3, 4, 5, 6, 7, 8\}$.
\end{ex}
% \textbf{Step 1.} Given a graph $\Gamma = (V, E)$ and a node $s_{0}$ as our input, we first check if any nodes satisfy Eq. (\ref{eq:ParityAssumption}), and thus violate the parity assumption as stated in Section IV.A in \cite{TBFS}. For every node that violates the assumption, we extend the graph $\Gamma$ to $\widehat{\Gamma}$, following the construction in Section IV.A in \cite{TBFS}, which attaches a triangular segment to the graph.
% \begin{figure}[H]
%     \centering
%     \includegraphics[width=\linewidth]{figures/step1.pdf}
%     \caption{Node $4$ in the initial graph with fixed vertex $0$ is the only node violating Eq.~\eqref{eq:ParityAssumption}. The resulting extended graph is shown on the right-hand side of the figure. All nodes in the extended graph satisfy the parity assumption.}
%     \label{fig:step1}
% \end{figure} \BT{Remove! We don't need this for trees. Renumber the remaining steps accordingly.}

\textbf{Step 1.} First, we execute $\texttt{Partition}(v)$ to obtain the initial partition~$\underset{v, \text{init}}{\text{\ScissorHollowRight}}$ by assigning each vertex~$u \in V \setminus \{v\}$ to its respective group~$C^{v}_{j, a}$. From this initial partition, we then identify the set of singleton vertices, $R = \{h \in V_{i} \mid \abs{V_{i}} = 1\}$ (see Figure \ref{fig:step1_partition}).
\begin{figure}[htbp]
    \centering
    \includegraphics[width=0.6\linewidth]{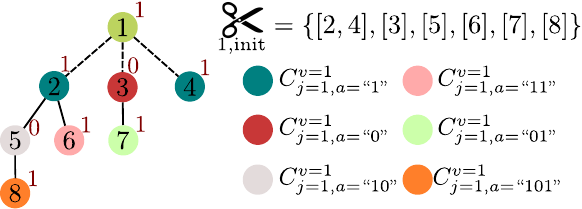}
    \caption{The initial grouping~${\underset{v=1, \text{init}}{\text{\ScissorHollowRight}}}$ obtained from \texttt{Partition}($1$), illustrating the assignment of nodes $\{2, 3, 4, 5, 6, 7, 8\}$ to their respective groups $C^{v=1}_{j, a}$. From the initial grouping, it is straightforward to obtain the singleton set $R = \{3, 5, 6, 7, 8\}$}
    \label{fig:step1_partition}
\end{figure}

\textbf{Step 2.} As the initial execution of $\texttt{Partition}(v)$ does not fully partition the vertices into singletons, we enter the refinement loop. We select a singleton vertex~$k \in R$ (i.e., a vertex whose respective cell~$C^{v}_{j, a}$ contains only $k$ itself) to serve as our next root. We then execute the subroutine $\texttt{Partition}(k)$, which yields the new partition~$\text{\ScissorHollowRight}_k$ and an updated singleton set~$R$ (see Figure \ref{fig:step2}).
\begin{figure}[htbp]
    %\vspace{-10mm}
    \centering
    \includegraphics[width=0.6\linewidth]{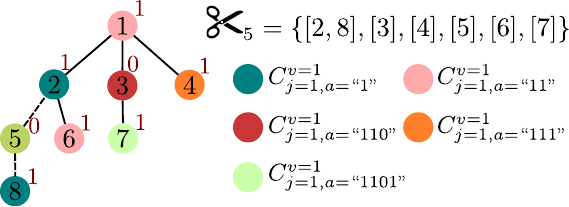}
    \caption{The grouping~$\text{\ScissorHollowRight}_5$ obtained from \texttt{Partition}($5$), illustrating the assignment of nodes $\{2, 3, 4, 5, 6, 7, 8\}$ to their respective groups $C^{v=5}_{j, a}$. From this grouping, we can update the singleton set to $R = \{3, 4, 5^*, 6, 7, 8\}$. Here, the asterisk for node $5$ denotes that it has been exploited as a root node.}
    \label{fig:step2}
\end{figure}

\textbf{Step 3.} Inside the refinement loop, after obtaining the new grouping $\text{\ScissorHollowRight}_k$ from \texttt{Partition}($k$), the subroutine $\texttt{Refinement}\left(\underset{v, \text{current}}{\text{\ScissorHollowRight}},\text{\ScissorHollowRight}_{k}\right)$ is performed to refine the initial partition. Given the updated partition~$\text{\ScissorHollowRight}_{k}$, the loop either (1) continues by rerooting unused singletons or (2) terminates when the partition forms a disjoint union of singleton sets or when all singletons have been used. Consequently, after exiting the loop, our program  returns the final partition~$\underset{v, \text{final}}{\text{\ScissorHollowRight}}$ (see Figure \ref{fig:step3}).
\begin{figure}[htbp]
    %\vspace{-10mm}
    \centering
    \includegraphics[width=0.6\linewidth]{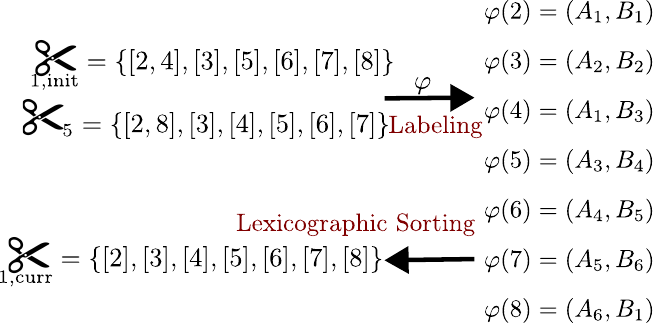}
    \caption{Given the obtained partition~$\text{\ScissorHollowRight}_5$, we update the initial partition~${\underset{1, \text{init}}{\text{\ScissorHollowRight}}}$ by calling subroutine $\texttt{Refinement}\left(\underset{1, \text{init}}{\text{\ScissorHollowRight}}, \text{\ScissorHollowRight}_5\right)$. Inside the subroutine, both partitions are first labeled using the string representation as defined in Eq.~\eqref{eq:CommonPartitionAndLexicogrSortingEquiv}. Then, lexicographic sorting is performed, which groups the nodes with the same labels. Since each vertex receives a unique labeling string, the resulting refinement is discrete; that is, every cell of the partition is a singleton.}
    \label{fig:step3}
\end{figure}

\begin{algorithm}
    \DontPrintSemicolon
    \caption{Parity-partition refinement}
    \label{alg:recursive-partition}
    \KwInput{Graph $\Gamma = (V,E)$ , fixed node $v \in V$}
    \KwOutput{partition $\underset{v, \text{final}}{\text{\ScissorHollowRight}}$}

    % $(D,A) \gets \mathrm{BFS}(\Gamma, s_{0})$\;
    % Check $D$ to obtain set of nodes $S'$ which violates the parity assumption in Equation \ref{eq:ParityAssumption}\;
    % \For{$u \in S'$}
    % {
    %     Extend $\Gamma$ to $\widehat{\Gamma}$ based on Section IV.A in \cite{TBFS} 
    % }
    %$\mathcal{T}\gets \emptyset$\ \tcp*{\small General Mapping}
    $\underset{v, \text{current}}{\text{\ScissorHollowRight}} \gets \texttt{Partition}(v)$ \tcp*{\small Vertex Partition}
    $S \gets \{v\}$ \tcp*{\small Utilized singleton nodes}
    $R \gets \left\{h \in V_{i} \mid \abs{V_{i}} = 1, V_{i} \in \underset{v, \text{current}}{\text{\ScissorHollowRight}} \right\}$ \; \tcp*{\small Singletons from initial partition}
    % \SetKwFunction{FixNode}{FixNode}
    % \SetKwProg{Fn}{Function}{:}{}
    % \Fn{\FixNode{$s$}}
    \While{$R \neq \emptyset$}
    {
        $w \gets R.\texttt{pop}()$\; 
        \If{$w \in S$}
        {
            \textbf{Continue}\;
        }
        $S \gets S \cup \{w\}$\;
        \tcp{Perform Partition}
        $\text{\ScissorHollowRight}_w\gets \texttt{Partition}(w)$\;
        % Add $\mathcal{T}_{s}$ to $\mathcal{T}$ \tcp*{\small Update mapping $\mathcal{T}$}
        \tcp{Perform Refinement}
        $\underset{v, \text{current}}{\text{\ScissorHollowRight}} \gets \texttt{Refinement}\left(\underset{v, \text{current}}{\text{\ScissorHollowRight}}, \text{\ScissorHollowRight}_w\right)$\;
        \tcp{Update singleton list}
        $R \gets R \cup \left\{h \in V_{i} \mid \abs{V_{i}} = 1, V_{i} \in \underset{v, \text{current}}{\text{\ScissorHollowRight}}, h \notin S \right\}$\;
        % \For{$u \in \widehat{\Gamma}(V) \setminus \{s\}$}
        % {
        %     \If{$|\mathcal{T}_{s}(u)|=1$}
        %     {
        %         $R\gets R\cup \{u\}$ \tcp*{\small Add singleton}
        %     }
        % }
        \If{
        $\abs{\underset{v, \text{current}}{\text{\ScissorHollowRight}}} = \abs{V} - 1$}
        {
            $\mathbf{Break}$\;
        }
        % $(\mathrm{done},\mathcal{X}_{s_0}) \gets \texttt{Matching}(\mathcal{T})$\;
        % \If{$\mathrm{done}=\mathrm{True}$}
        % {
        %     \Return $(\mathrm{True},\mathcal{X}_{s_0})$\;
        % }
        % \For{$u\in R$}
        % {
        %     \If{$u\notin S$}
        %     {
        %         \tcp{\small Enter recursion}
        %         $(\mathrm{done},\mathcal{X}_{s_0})\gets \FixNode{u}$\;
        %         \If{$\mathrm{done} = \mathrm{True}$}
        %         {
        %             \Return $(\mathrm{True},\mathcal{X}_{s_0})$
        %         }
                
        %     }
        % }
        % \tcp{If not fully partitioned after fully iterating through all nodes}
        % \If{$\abs{S} = \abs{\widehat{\Gamma}(V)} - 1$}
        % {
        %     \Return $(\mathrm{False}, \mathcal{X}_{s_0})$
        % }
        % \Return $(\mathrm{done},\mathcal{X}_{s_0})$\;
    }
    \Return $\underset{v, \text{current}}{\text{\ScissorHollowRight}}$\;
\end{algorithm}

\begin{rmk}\label{rmk:Optimality}
While the partitions produced by Algorithm~\ref{alg:recursive-partition} yield highly effective structural decompositions, as demonstrated by the empirical findings in Section~\ref{sec:NumResults}, the theoretical results in~\cite{TBFS} do not preclude the existence of linear combinations of Pauli-$X$ operators with smaller support inside the standard reduced dynamical Lie algebra. Consequently, we do not claim global optimality for the resulting partition sizes. 
\end{rmk}

% \subsection{Complexity Analysis}
% \label{sec:complexity_analysis}

\section{Numerical Results}
\label{sec:NumResults}
In this section, we investigate the partition structure of the vertices in random geodetic graphs produced by Algorithm~\ref{alg:recursive-partition} (see Section~\ref{sec:Method}). Specifically, the vertices are partitioned according to whether the corresponding uniform linear combinations of Pauli $X$ operators belong to the standard reduced DLA $\mathfrak{g}^{\,v}_{\Gamma,\mathrm{std}}$ (see Appendix~\ref{sec:groundwork}). Our experimental dataset consists of $1000$ random acyclic graphs generated from uniformly random Prüfer sequences, with sizes ranging from $20$ to $1000$ vertices. In all of our experiments, all $n$ choices of fixed vertices are evaluated. We begin by introducing several preliminary definitions and results.
% Unless otherwise specified, these are Erdős–Rényi graphs generated with an edge-creation probability of $p = \frac{\log n}{n}$, where $n$ denotes the number of nodes. This parameterization ensures that the generated graphs are sparse~\cite{erdHos1960evolution}.
%thereby allowing for a diversified path parity-degree profiles.
% Importantly, MaxCut remains NP-hard even when restricted to sparse graph classes (see~\cite[Theorem~13]{Yan1978} and the references therein).

\subsection{Symmetry Invariance and Vertex Splittability}

To rigorously characterize the structure of the reduced dynamical Lie algebras, we exploit  the symmetry properties of the underlying graph.

Let $\Gamma=(V,E)$ be a connected graph on $n$ vertices, and let $\mathrm{Aut}(\Gamma)$ denote its automorphism group, acting on the Hilbert space $W=(\mathbb{C}^2)^{\otimes n}$ by permuting the tensor factors. 
For any designated vertex $v\in V$, the standard reduced dynamical Lie algebra $\mathfrak{g}^{v}_{\Gamma,\mathrm{std}}$ is invariant under the action of the stabilizer subgroup:
\begin{equation}
\mathrm{Stab}(v) := \{\sigma \in \mathrm{Aut}(\Gamma) : \sigma(v)=v\}.
\end{equation}
This structural result appears as part of Lemma~\ref{lem:DLAinvariance}. 

The interaction between the stabilizer subgroup and the recursive refinement of the vertex set motivates our core classification framework. We formalize this via the following algorithmic partition properties.

\begin{defn}
\label{defn:XgateSplit}
Let $\Gamma = (V,E)$ be a graph and let $v \in V$ be a fixed target vertex. 
Define an equivalence relation $\sim_v$ on $V$ by
\[
u \sim_v w \quad\Longleftrightarrow\quad \exists\,\sigma \in \mathrm{Stab}(v) \text{ such that } \sigma(u) = w.
\]

The pair $(\Gamma, v)$ is said to be:
\begin{itemize}
    \item \emph{orbit-coherent} if the cells of the vertex partition $\underset{v}{\text{\ScissorHollowRight}}$ (see Eq.~\eqref{eq:Partition}) coincide exactly with the orbits of $\mathrm{Stab}(v)$ on $V \setminus \{v\}$. That is, the partition matches the quotient set:
    \begin{equation}
        \label{eq:Xgate_attainable_partition}
        \underset{v,\max}{\text{\ScissorHollowRight}} = (V \setminus \{v\}) / \sim_v.
    \end{equation}
    
    \item \emph{fully resolved} if the resulting vertex partition $\underset{v}{\text{\ScissorHollowRight}}$ refines completely into singletons.
\end{itemize}  

In other words, Algorithm~\ref{alg:recursive-partition} achieves full symmetry reduction when $(\Gamma, v)$ is orbit-coherent, and achieves maximal symmetry breaking (complete vertex identification) when $(\Gamma, v)$ is fully resolved, meaning every vertex is uniquely distinguished.
\end{defn}

\subsection{Performance Metrics and Statistics}

We now introduce several quantitative metrics to measure how closely the output of Algorithm~\ref{alg:recursive-partition} approximates full Pauli-$X$ splitting (see Definition~\ref{defn:XgateSplit}), particularly in scenarios where the theoretical maximum splitting is not achieved.

\begin{defn}
\label{defn:SplittingStatistics}
Let $\mathcal{S}$ be a sample drawn from a family of geodetic graphs on $n$ nodes. 

\begin{itemize}
    \item \textbf{Vertex-level statistics.} For a graph $\Gamma\in \mathcal{S}$ and a fixed distinguished vertex $v\in V$, let $k_v$ denote the number of subsets in the vertex partition in Eq.~\eqref{eq:Partition} produced by Algorithm~\ref{alg:recursive-partition}. We define the average subset size $\mathcal{A}_{(\Gamma,v)}$ and the number of singleton subsets $\mathcal{B}_{(\Gamma,v)}$ by
    \begin{equation}
    \label{eq:avg_subset_size_vertex}
    \begin{aligned}
        \mathcal{A}_{(\Gamma,v)} &:= \frac{n-1}{k_v}, \\
        \mathcal{B}_{(\Gamma,v)} &:= \#\{V_i \mid |V_i|=1\}.
    \end{aligned}
    \end{equation}
    Thus, $\mathcal{A}_{(\Gamma,v)}$ measures the average size of the resulting algorithmic components, while $\mathcal{B}_{(\Gamma,v)}$ tracks how many vertices are completely isolated into singleton classes.

    \item \textbf{Graph-level statistics.} For a fixed graph $\Gamma$, we average these vertex-level quantities over all choices of the distinguished vertex $v$:
    \begin{equation}
    \label{eq:avg_subset_size_graph}
    \begin{aligned}
        \mathcal{A}_{\Gamma} &:= \frac{1}{n}\sum_{v\in V} \mathcal{A}_{(\Gamma,v)}, \\
        \mathcal{B}_{\Gamma} &:= \frac{1}{n}\sum_{v\in V} \mathcal{B}_{(\Gamma,v)}.
    \end{aligned}
    \end{equation}

    \item \textbf{Sample-level statistics.} Averaging the graph-level statistics over the entire sample $\mathcal{S}$ yields
    \begin{equation}
    \label{eq:avg_subset_size_sample}
    \begin{aligned}
         \mathcal{A}_\mathcal{S} &:= \frac{1}{|\mathcal{S}|}\sum_{\Gamma\in \mathcal{S}} \mathcal{A}_{\Gamma}, \\
         \mathcal{B}_\mathcal{S} &:= \frac{1}{|\mathcal{S}|}\sum_{\Gamma\in \mathcal{S}} \mathcal{B}_{\Gamma}.
    \end{aligned}
    \end{equation}

    \item \textbf{Maximal attainable benchmarks.} To gauge the performance of the algorithm against the theoretical upper bound, we introduce analogous metrics for the maximal splitting partition. Let $\underset{v,\max}{\text{\ScissorHollowRight}}$ be the partition of vertices determined by the orbits of the stabilizer of $v$ (see Eq.~\eqref{eq:Xgate_attainable_partition}), and let $m_v$ be the number of subsets in this partition. We define the vertex, graph, and sample-level benchmarks as:
    \begin{equation}
    \label{eq:avg_subset_size_sampleMax}
    \begin{aligned}
        \mathcal{A}_{(\Gamma,v,\text{max})} &:= \frac{n-1}{m_v},\\ \mathcal{B}_{(\Gamma,v,\text{max})} &:= \#\{V_i \mid |V_i|=1\}, \\
        \mathcal{A}_{\Gamma,\text{max}} &:= \frac{1}{n}\sum_{v\in V} \mathcal{A}_{(\Gamma,v,\text{max})}, \\ \mathcal{B}_{\Gamma,\text{max}} &:= \frac{1}{n}\sum_{v\in V} \mathcal{B}_{(\Gamma,v,\text{max})}, \\
        \mathcal{A}_{\mathcal{S},\text{max}} &:= \frac{1}{|\mathcal{S}|}\sum_{\Gamma\in \mathcal{S}} \mathcal{A}_{\Gamma,\text{max}}, \\
        \mathcal{B}_{\mathcal{S},\text{max}} &:= \frac{1}{|\mathcal{S}|}\sum_{\Gamma\in \mathcal{S}} \mathcal{B}_{\Gamma,\text{max}}.
    \end{aligned}
    \end{equation}
\end{itemize}
\end{defn}
These metrics provide a global quantitative assessment of the partition refinement achieved by Algorithm~\ref{alg:recursive-partition}. Specifically, smaller values of $\mathcal{A}_{\mathcal{S}}$ and larger values of $\mathcal{B}_{\mathcal{S}}$ indicate that the algorithm's output closely approaches the optimal refinement benchmarks (through group automorphism), $\mathcal{A}_{\mathcal{S},\text{max}}$ and $\mathcal{B}_{\mathcal{S},\text{max}}$.

The resulting statistics for Erd\H{o}s-R\'enyi acyclic graph samples appear in Tables~\ref{table:SplittingStatisticsTable} and~\ref{table:SplittingStatisticsTable2}. Furthermore, we 
also consider non-acyclic geodetic graphs generated as in~\cite{stober2023geodetic}. Specifically, we look at the infinite families of geodetic graphs created by joining two odd cycles, $F_{k}$ described in Section $6$ of~\cite{stober2023geodetic}. In that construction, a new graph is created with $\abs{V(F_{k})} = 2k + 13$ vertices and $\abs{E(F_{k})} = 2k + 18$. For the numerical benchmark, we look at ranging $k \in \{5, 15, 25, \dots, 175 \}$ as shown in Table~\ref{table:geodetic_Fk}, here we show the values of $\mathcal{A}_{\Gamma}$, $\mathcal{B}_{\Gamma}$ for each $\Gamma = F_{k}$.

\begin{rmk}
\label{rmk:splitting_statistics}
The results in Table~\ref{table:SplittingStatisticsTable} indicate that the average subset size, $\mathcal{A}_{\mathcal{S}}$, stabilizes around $1.27$. Meanwhile, the average number of singleton cells in the resulting partition, $\mathcal{B}_{\mathcal{S}}$, exceeds $64\%$ of the total number of vertices.
\end{rmk}

\begin{rmk}
\label{rmk:splitting_statistics_non_acyclic}
The results in Table~\ref{table:geodetic_Fk} indicate that the average size of the partition for each graph over different vertex fixings $\Gamma = F_{k}$, $\mathcal{A}_{\Gamma}$, decreases as $k$ increases and reaches the value of $1.344$. Meanwhile, the average number of singletons in the resulting partition for each graph over different vertex fixings, $\mathcal{B}_{\Gamma}$, is around $50\%$ of the total number of vertices.
\end{rmk}

\begin{table*}%[H]
\centering
\resizebox{1.0\linewidth}{!}{
\begin{tabular}{|c|c|c|c|c|c|c|c|c|c|c|c|c|c|c|c|c|c|c|}
\hline
$n$ & $20$ & $30$ & $40$ & $50$ & $60$ & $70$ & $80$ & $90$ & $100$ & $200$ & $300$ & $400$ & $500$ & $600$ & $700$ & $800$ & $900$ & $1000$\\
\hline

$\mathcal{A}_{\mathcal{S}}$ &
1.341 &
1.318 &
1.302 &
1.290 &
1.286 &
1.279 &
1.274 &
1.278 &
1.276 &
1.272 &
1.269 &
1.270 & 1.271 & 1.271  & 1.269 & 1.270 & 1.268 & 1.270\\
\hline
$\mathcal{B}_{\mathcal{S}}$
& 11.03 & 17.38 & 24.00 & 30.56 & 37.12 & 43.37 & 50.15 & 56.94 & 62.86 & 126.87 & 191.97 & 256.82 & 321.08 & 385.57 & 450.01 & 514.77 & 579.65 & 643.16\\
\hline
% $\mathcal{D}_{\mathcal{S}}$
% &  &  &  &  &  &  & &  &  & & &  &  &  &  &  &  & \\
% \hline
\end{tabular}
}
\caption{Average refinement statistics for partitions generated by Algorithm~\ref{alg:recursive-partition} across $1,000$ random connected acyclic graph samples with $n = 20$ to $1,000$ vertices.}
\label{table:SplittingStatisticsTable}
\end{table*}

\begin{table*}%[H]
\centering
\resizebox{1.0\linewidth}{!}{
\begin{tabular}{|c|c|c|c|c|c|c|c|c|c|c|c|c|c|c|c|c|c|c|}
\hline
$n$ & $20$ & $30$ & $40$ & $50$ & $60$ & $70$ & $80$ & $90$ & $100$ & $200$ & $300$ & $400$ & $500$ & $600$ & $700$ & $800$ & $900$ & $1000$\\
\hline

$\mathcal{A}_{\text{max}}/\mathcal{A}_{\mathcal{S}} $ &
0.874 &
0.887 &
0.902 &
0.906 &
0.913 &
0.910 &
0.916 &
0.920 &
0.915 &
0.917 &
0.918 &
0.921 & 0.921 & 0.922  & 0.921 & 0.922 & 0.922 & 0.921\\
\hline
$\mathcal{B}_{\mathcal{S}}/ \mathcal{B}_{\text{max}}$
& 0.936 & 0.943 & 0.950 & 0.952 &  0.955 & 0.953 & 0.956 & 0.958 & 0.955 & 0.956 & 0.957 & 0.958 & 0.958 & 0.958 & 0.958 & 0.958 & 0.958 & 0.959\\
\hline
% $\mathcal{D}_{\mathcal{S}}$
% &  &  &  &  &  &  & &  &  & & &  &  &  &  &  &  & \\
% \hline
\end{tabular}
}
\caption{Average ratios of refinement statistics for partitions generated by Algorithm~\ref{alg:recursive-partition} relative to optimal partitions under automorphism group actions. Results are based on $1,000$ random connected acyclic graph samples with $n = 20$ to $1,000$ vertices.}
\label{table:SplittingStatisticsTable2}
\end{table*}

\begin{table*}%[H]
\centering
\resizebox{1.0\linewidth}{!}{
\begin{tabular}{|c|c|c|c|c|c|c|c|c|c|c|c|c|c|c|c|c|c|c|}
\hline
$k$ & $5$ & $15$ & $25$ & $35$ & $45$ & $55$ & $65$ & $75$ & $85$ & $95$ & $105$ & $115$ & $125$ & $135$ & $145$ & $155$ & $165$ & $175$\\
\hline

$\mathcal{A}_{\Gamma}$ &
1.543 &  1.423 &  1.395 &  1.381 &  1.372 & 1.365 &  1.361 &  1.358 &  1.355 & 1.353&  1.351&  1.35& 
 1.348 & 1.347&  1.346 &  1.346 &  1.345 &  1.344\\
\hline
$\mathcal{B}_{\Gamma}$
& 11.826 &
21.302 &
31.048 &
40.916 &
50.835 &
60.78 &
70.741 &
80.712 &
90.689 &
100.67 &
110.655 &
120.642 &
130.631 &
140.622 &
150.614 &
160.607 &
170.601 &
180.595 \\
\hline
\end{tabular}
}
\caption{Average refinement statistics for partitions generated by Algorithm~\ref{alg:recursive-partition} across $18$ connected geodetic graph with $k \in \{5, 15, 25 \dots, 175\}$. Here the first graph $F_{k=5}$ has $23$ vertices and $28$ edges while the largest graph $F_{k=175}$ has $363$ vertices and $368$ edges.}
\label{table:geodetic_Fk}
\end{table*}

\section{Practical Implications}
\label{sec: implications}
We now examine the implications of Algorithm~\ref{alg:recursive-partition},
introduced in this work, for the structure of dynamical Lie algebras.
\begin{defn}
    \label{defn:reduced_graph}
    The \emph{reduced graph} $\Gamma_v$ is the subgraph obtained from $\Gamma$ by removing the vertex $v$ along with all its incident edges.
\end{defn}
Our analysis relies on the containment of the free DLA of the reduced
graph in the corresponding free reduced DLA,
\begin{equation}
\label{eq:reduced_dla_containment}
\mathfrak{g}_{\Gamma_v,\mathrm{free}}
\subseteq
\mathfrak{g}^v_{\Gamma,\mathrm{free}},
\end{equation}
together with the following direct-sum decomposition. If the reduced
graph has $m$ connected components,
\[
\Gamma_v=\Gamma_1\sqcup\Gamma_2\sqcup\cdots\sqcup\Gamma_m,
\]
then, directly from the definition of the free reduced DLA,
\begin{equation}
\label{eq:free_dla_direct_sum}
\mathfrak{g}^v_{\Gamma,\mathrm{free}}
=
\mathfrak{g}^v_{\widehat{\Gamma}_1,\mathrm{free}}
\oplus
\mathfrak{g}^v_{\widehat{\Gamma}_2,\mathrm{free}}
\oplus\cdots\oplus
\mathfrak{g}^v_{\widehat{\Gamma}_m,\mathrm{free}},
\end{equation}
where $\widehat{\Gamma}_j$ is the subgraph obtained by augmenting $\Gamma_j$ with $v$ through the unique edge $(v,v_j)$ in $\Gamma$ connecting $v$ to $\Gamma_j$. Here, $v_j$ denotes the vertex in $\Gamma_j$ adjacent to $v$. We then combine
\eqref{eq:reduced_dla_containment}--\eqref{eq:free_dla_direct_sum}
with the classification of the Lie algebras appearing as the direct-sum
constituents in \eqref{eq:free_dla_direct_sum} to characterize the
expressivity of free reduced parameterized QAOA circuits on tree graphs
within the corresponding reduced Hilbert spaces.

The following result is essentially a reformulation of Proposition $26$
in~\cite{KLFCCZ}.
\begin{thm}
\label{thm:FreeDLAExtensionBySingleZ}
Let $\Gamma = (V_1 \sqcup V_2, E)$ be a connected bipartite graph that is not a path graph, and let $v \in V_1$ be a vertex.  Define the extension of the free DLA $\mathfrak{g}_{\Gamma,\mathrm{free}}$ by the additional generator $iZ_v$ as
\begin{equation}
\label{eq:FreeDLAExtensionBySingleZDefinition}
\operatorname{Ext}_v
(\mathfrak{g}_{\Gamma,\mathrm{free}})
:=
\left\langle
\mathfrak{g}_{\Gamma,\mathrm{free}},\, iZ_v
\right\rangle_{\mathrm{Lie}}.
\end{equation}
Then
\begin{equation}
\label{eq:FreeDLAExtensionBySingleZClassification}
\operatorname{Ext}_v
(\mathfrak{g}_{\Gamma,\mathrm{free}})
\simeq
\begin{cases}
\mathfrak{sp}\!\left(2^n\right),
&
v\in V_1 \text{ and } |V_1| \text{ is odd},
\\[2mm]
\mathfrak{so}\!\left(2^n\right),
&
 \text{ otherwise}.
\end{cases}
\end{equation}
\end{thm}

\begin{rmk}
\label{rmk:InvariantBilinearFormSingleZExtension}
Consider the Pauli string
\begin{equation}
\label{eq:SvDefinition}
S_v:=
\left(\prod_{u\in V_1}Y_u\right)
\left(\prod_{u\in V_2}Z_u\right).
\end{equation}
Since $S_v$ is a Pauli string, it is invertible and therefore defines
a nondegenerate complex bilinear form on
$W=(\mathbb C^2)^{\otimes n}$ by
\begin{equation}
\label{eq:BilinearFormSv}
\beta_v(\psi,\phi)
:=
\psi^T S_v\phi.
\end{equation}
Its symmetry type is determined by the parity of $|V_1|$. Indeed,
using $Y^T=-Y$ and $Z^T=Z$, we obtain
\begin{equation}
\label{eq:SvTranspose}
S_v^T
=
(-1)^{|V_1|}S_v.
\end{equation}
Hence $\beta_v$ is symmetric when $|V_1|$ is even and skew-symmetric when $|V_1|$ is odd. This is precisely the invariant bilinear form preserved by $\operatorname{Ext}_v(\mathfrak{g}_{\Gamma,\mathrm{free}})$,
corresponding to the orthogonal and symplectic cases, respectively, in Theorem~\ref{thm:FreeDLAExtensionBySingleZ}.
\end{rmk}

\begin{rmk}
\label{rmk:TwoVertexExtension}
Let $\Gamma$ be the graph on two vertices $v,w$ with a single edge
$(v,w)$. Then, by Proposition~18 in~\cite{KLFCCZ},
\begin{equation}
\label{eq:TwoVertexFreeDLAExtension}
\operatorname{Ext}_v
\bigl(\mathfrak{g}_{\Gamma,\mathrm{free}}\bigr)
\simeq
\mathfrak{so}(5)\simeq \mathfrak{sp}(4).
\end{equation}

Moreover, the corresponding connected Lie group, $Sp(4)$, acts transitively on
the unit sphere of the two-qubit Hilbert space $W$, up to a global phase.
\end{rmk}

Combining Algorithm~\ref{alg:recursive-partition} with these observations and the free DLA classification from~\cite{KLFCCZ} yields several controllability and expressivity results for reduced multi-angle QAOA ansätze, as detailed below.

\subsection{Fully Resolved Pairs and Pure-State Controllability}

Whenever a given pair $(\Gamma,v)$ is fully resolved (see
Definition~\ref{defn:XgateSplit}), the standard reduced DLA coincides
with the free reduced DLA. Furthermore, since $\Gamma$ is a tree, each connected component
$\Gamma_j$, $j=1,\dots,m$, of the reduced graph $\Gamma_v$ is bipartite
and is connected to $v$ by a single edge; otherwise, $\Gamma$ would
contain a cycle.

Consequently,
\[
\mathfrak{g}^v_{\widehat{\Gamma}_j,\mathrm{free}}
=
\operatorname{Ext}_{v_j}
\bigl(\mathfrak{g}_{\Gamma_j,\mathrm{free}}\bigr),
\]
and Eq.~\eqref{eq:free_dla_direct_sum} can be written as
\begin{equation}
\label{eq:FreeReducedDLAExtensionDirectSum}
\mathfrak{g}^v_{\Gamma,\mathrm{free}}
=
\operatorname{Ext}_{v_1}
\bigl(\mathfrak{g}_{\Gamma_1,\mathrm{free}}\bigr)
\oplus
\operatorname{Ext}_{v_2}
\bigl(\mathfrak{g}_{\Gamma_2,\mathrm{free}}\bigr)
\oplus\cdots\oplus
\operatorname{Ext}_{v_m}
\bigl(\mathfrak{g}_{\Gamma_m,\mathrm{free}}\bigr).
\end{equation}

Provided that no connected component $\Gamma_j$ is a path graph on more than two vertices, Theorem~\ref{thm:FreeDLAExtensionBySingleZ} and Remark~\ref{rmk:TwoVertexExtension} imply that each extended Lie algebra $\operatorname{Ext}_v\bigl(\mathfrak{g}_{\Gamma_j,\mathrm{free}}\bigr)$ is isomorphic to either $\mathfrak{so}(2^{|\Gamma_j|})$ or $\mathfrak{sp}(2^{|\Gamma_j|})$, which correspond respectively to the orthogonal and symplectic Lie groups, $\mathrm{SO}(2^{|\Gamma_j|})$ and $\mathrm{Sp}(2^{|\Gamma_j|})$. Recall that in the latter case, the corresponding Lie group action acts transitively (up to a global phase) on the unit sphere of the subspace $W_{\Gamma_j}$, whereas in the former case it does not (see, e.g., Theorem~2 in~\cite{AD}).

%Let $\ket{t}, \ket{h} \in \bigotimes\limits_{j=1}^m W_{\Gamma_j}$ be normalized states. The associated connected Lie group corresponding to the Lie algebra $\operatorname{Ext}_v\bigl(\mathfrak{g}_{\Gamma_j,\mathrm{free}}\bigr)$ acts transitively on the  unit sphere $\mathcal{S}(W_{\Gamma_j})$ (up to global phase).

Consequently, consider any two states $\ket{t}$ and $\ket{h}$ which both lie within the product subspace
\begin{equation}
\label{eq:cartesian_product_spheres}
\bigotimes_{j=1}^m W_{\Gamma_j}
\end{equation}
such that each tensor factor has unit norm, and the states differ only in factors corresponding to subgraphs $\Gamma_j$ for which the vertex set containing $v$ in the extended bipartite graph $\widehat{\Gamma}_j$ has odd cardinality. It follows from~\cite{dalessandro2009} that there exists a positive integer $L \in \mathbb{Z}_{>0}$ such that an $L$-layer parameterized reduced QAOA circuit $U(\boldsymbol{\beta},\boldsymbol{\gamma})$ (defined in Eq.~\eqref{qaoa-chain}) achieves exact state preparation up to a global phase $\phi$:
\begin{equation}
\label{eq:exact_state_prep_xgate}
U(\boldsymbol{\beta},\boldsymbol{\gamma})\ket{h} = e^{i\phi}\ket{t}.
\end{equation}
When the standard and free reduced DLAs coincide, this state preparation guarantee holds for the standard reduced QAOA circuit as well.

\begin{rmk}
When $v$ is a leaf, the reduced graph $\Gamma_v$ is connected ($m=1$) and bipartite. Let $v'$ denote the unique neighbor of $v$ in $\Gamma$, and let $V_{v'}$ be the bipartite partition block of $\Gamma_v$ containing $v'$. If $\Gamma_v$ is not a path graph on more than two vertices and $|V_{v'}|$ is odd, the associated dynamical Lie group acts transitively (up to a global phase) on the unit sphere of the reduced Hilbert space $W_v$. In quantum control theory, this transitivity corresponds to pure-state controllability on $W_v$.
\end{rmk}

\subsection{Partial Splitting and Subsystem Controllability}

Even when a pair $(\Gamma,v)$ is not fully resolved, the partial partition structure still yields non-trivial control algebraic guarantees. Specifically, consider the subset of singleton vertices identified by the partition for a given vertex $v$:
\begin{equation}
    \label{eq:subset_of_singletons}
    \mathcal{S}_{\Gamma,v}:=\{w\in V\setminus\{v\}\mid iX_w \in \mathfrak{g}^v_{\Gamma,\mathrm{std}}\}.
\end{equation}
For any connected subgraph $\Gamma' \subset \Gamma$ supported on $\mathcal{S}_{\Gamma,v}$, the nested commutator relation from Eq.~\eqref{eq:RecoverZZ2} ensures the DLA containment:
\begin{equation}
    \label{eq:partial_inclusion}
    \mathfrak{g}_{\Gamma',\mathrm{free}} \subseteq \mathfrak{g}^v_{\Gamma,\mathrm{std}}.
\end{equation}
This containment immediately yields a rigorous lower bound on the dimension of the standard reduced DLA:
\begin{equation}
    \label{eq:partial_lower_bound}
    \dim(\mathfrak{g}^v_{\Gamma,\mathrm{std}}) \ge \dim(\mathfrak{g}_{\Gamma',\mathrm{free}}).
\end{equation}

Beyond dimension bounds, the embedding established in Eq.~\eqref{eq:partial_inclusion} guarantees subsystem controllability whenever the localized free dynamical Lie algebra acts transitively on the target subsystem Hilbert space. %As shown in the preceding subsection, transitivity (up to a global phase) on the unit sphere of each invariant parity subspace $W_{\Gamma',\mathrm{even}}$ and $W_{\Gamma',\mathrm{odd}}$ holds precisely when the bipartite vertex partition of $\Gamma'$ is not of even-even parity type.

Now consider any two normalized states $\ket{t}, \ket{h} \in W_v$ that differ only on the qubits in $\Gamma'$:
\begin{equation}
\label{eq:StateDecomposition}
\ket{t} = \ket{t_{\Gamma'}} \otimes \ket{\psi_{\setminus \Gamma'}}, \quad \ket{h} = \ket{h_{\Gamma'}} \otimes \ket{\psi_{\setminus \Gamma'}}.
\end{equation}
Note that as a subgraph of an acyclic graph, $\Gamma'$ is necessarily bipartite with vertex partition $V(\Gamma') = V_1 \sqcup V_2$. If it is not a path on more than two nodes and the parity pair $(|V_1|, |V_2|)$ is not even-even, then $\mathfrak{g}_{\Gamma',\mathrm{free}}$ is isomorphic to either $\mathfrak{sp}(W_{\Gamma',\mathrm{even}}) \oplus \mathfrak{sp}(W_{\Gamma',\mathrm{odd}})$ (if $|V_1|$ and $|V_2|$ are both odd) or $\mathfrak{su}(2^{|V(\Gamma')|-1})$ acting diagonally on the $W_{\Gamma',\mathrm{even}}$ and $W_{\Gamma',\mathrm{odd}}$ sectors (if $|V(\Gamma')|$ is odd); see Theorem~1 in \cite{KLFCCZ}. It follows that the associated connected Lie group $G_{\Gamma'} = \exp(\mathfrak{g}_{\Gamma',\mathrm{free}})$ acts transitively (up to a global phase) on the unit spheres inside the invariant subspaces $W_{\Gamma',\mathrm{even}}$ and $W_{\Gamma',\mathrm{odd}}$.

Provided the bipartite partition of $\Gamma'$ is not even-even and the subsystem states $\ket{t_{\Gamma'}}$ and $\ket{h_{\Gamma'}}$ lie within the same parity subspace ($W_{\Gamma',\mathrm{even}}$ or $W_{\Gamma',\mathrm{odd}}$), this transitivity property guarantees \cite{dalessandro2009} the existence of a layer depth $L \in \mathbb{Z}_{>0}$ such that an $L$-layer reduced QAOA circuit $U(\boldsymbol{\beta},\boldsymbol{\gamma})$ achieves exact subsystem state preparation up to a global phase $\phi$:
\begin{equation}
\label{eq:approximation_bound_phase}
U(\boldsymbol{\beta},\boldsymbol{\gamma})\ket{h} = e^{i\phi}\ket{t}.
\end{equation}

Although a pair $(\Gamma, v)$ is rarely fully resolved in large graphs, partial splittings are typically extensive. For example, in $1000$-vertex tree graphs, the singleton subset $\mathcal{S}_{\Gamma,v}$ contains over $600$ vertices on average. Consequently, the standard reduced DLA contains the free DLA of a large subgraph $\Gamma'$, enabling standard reduced QAOA to inherit subsystem controllability over a significant fraction of the qubits while leaving spectator qubits invariant.

\section{Conclusion}

In this work, we investigated the dynamical Lie algebras associated with the reduced variant of the Quantum Approximate Optimization Algorithm for the MaxCut problem on connected geodetic graphs. While the hardness of MaxCut on geodetic graphs is unknown and is trivial when we consider the subclass of acyclic graphs, characterizing the precise quantum dynamics and expressivity of the corresponding ansatz remains a non-trivial and critical task. We focused specifically on reduced QAOA architectures arising from lossless reductions of the problem space, achieved by fixing a single bit value. Building upon the framework introduced in~\cite{TBFS}, we developed a recursive vertex parity-partition refinement algorithm (Algorithm~\ref{alg:recursive-partition}). This algorithm characterizes the structural aspects of reduced DLAs without requiring the computationally expensive explicit construction of the Lie closure, achieving an efficient worst-case time complexity of $\mathcal{O}\bigl(|V|\bigl(\abs{V} + \abs{E}\bigr)\bigr)$.

In cases where Algorithm~\ref{alg:recursive-partition} completely refines the vertex set into singleton subsets, we establish that the reduced DLA coincides with the Lie algebra arising from the multi-angle QAOA (the free ansatz). %Under these conditions, the structure and dimensionality of the DLA can be directly inferred from the classification of the latter. 
Conversely, when full splitting does not occur, the partitioning procedure still yields significant structural insights by identifying specific vertices whose associated Pauli-$X$ terms belong to the reduced DLA. This induces subalgebra inclusions of the form as in Eq.~\eqref{eq:partial_inclusion}, providing explicit lower bounds on the total dimension of the reduced DLA. More precisely, the singleton ($1$-vertex) subsets in the final partition imply the inclusion of the free Lie algebras supported on those vertices. This capability is particularly powerful, as it allows us to systematically derive lower dimension bounds in a straightforward manner (cf.~Theorem~1 in~\cite{KLFCCZ}). Consequently, even partial splitting provides valuable, rigorous information regarding the algebraic structure and ultimate expressivity of the reduced QAOA ansatz.

%Our numerical results on acyclic sparse Erdős--Rényi graphs $\Gamma(n,p)$ with $p=\frac{\log n}{n}$ show that, for a large fraction of vertices $v \in V$, the pair ($\Gamma, v$) is $X$-gate splittable. In these cases, the induced partition refines into singleton sets, implying that the equality \eqref{eq:EqualityOfReducedDLAs} holds true. \BT{Fix!}

%Furthermore, when a pair $(\Gamma, v)$ is not $X$-gate splittable, the partition still identifies vertices for which the corresponding generators $iX_{w}$ lie in $g_{\Gamma,\text{std}}^{v}$. This leads to the inclusion described in Eq.~\ref{eq:partial_inclusion}, for suitable subgraphs $\Gamma' \subset \Gamma$. As a result, we can obtain lower bounds on the dimensions of the reduced DLA, even in the absence of full splitting.

%Under mild connectivity assumptions, this further implies that every vector in the reduced Hilbert space can be reached from the initial state under the action of the standard reduced DLA. Correspondingly, every state (i.e., unit vector in that Hilbert space) can be obtained via the associated dynamical Lie group. 
Additionally, the framework presented here carries significant algorithmic implications. Specifically, in the regime of large QAOA depth, the DLA dimension is closely linked to the variance of the QAOA loss landscape. It thus serves as a proxy for fundamental landscape properties, including expressivity and the potential onset of barren plateaus~\cite{FHCKYHSP,LJGCC,LCSMCC,RBSKMLC}.

More broadly, our algorithmic framework applies beyond geodetic graphs (Remark~\ref{rmk:TreesSatisfyParityAssumption}) and can be extended to any graph class fulfilling the hypotheses of Remark~\ref{rmk:TreesSatisfyParityAssumption} and Lemma~\ref{lem:refinement_closure}. Applying Algorithm~\ref{alg:recursive-partition} to deduce the structural properties of dynamical Lie algebras on these broader graph topologies remains a promising direction for future work.
%Overall, this work demonstrates that symmetry-induced reductions, when viewed through a Lie-algebraic lens, provide a powerful and nuanced tool for understanding and designing QAOA circuits with good trainability-expressivity tradeoffs.

\section*{Acknowledgment}
This research was supported in part by NSF Grant \#2427042.

\bibliographystyle{unsrtnat}
\bibliography{References}

\appendix

\section{Groundwork}
\label{sec:groundwork}

In this appendix, we recall the structural ingredients underlying our analysis of reduced dynamical Lie algebras, as introduced in~\cite{TBFS}. In particular, we describe how graph-theoretic features can be used to identify \emph{distinguished elements} in the Lie algebra, which form the basis for our computational methods.

\textbf{Distance Layers and Parity-Degree Sequences.} Let $\Gamma = (V, E)$ be a connected graph, and let $v \in V$ be a distinguished vertex with a fixed associated bit value, serving as the root of the construction. 

\begin{defn}
\label{defn:vertexDistSubsets}
The distance layers of $\Gamma$ with respect to $v$ are defined by
\begin{equation}
    \label{eq:distLayers}
    \mathcal{N}_{v,j} := \{ w \in V \mid \mathrm{dist}(v,w) = j \}, \quad j \geq 1.
\end{equation}
\end{defn}

These layers partition the vertex set $V$ according to the shortest path distance from the fixed vertex $v$.

A further refinement of the distance layers $\mathcal{N}_{v,j}$ can be obtained according to the parity patterns of vertex degrees along shortest paths to the fixed node. 

\begin{defn}
\label{defn:vertexSubsets}
Let $v$ be a vertex in a connected graph $\Gamma$, and let $\mathbf{a} := (a_1, a_2, \dots, a_j) \in \mathbb{Z}_2^j$ be a binary sequence of length $j$. The subset $\mathcal{C}_{j, \mathbf{a}}^v \subseteq \mathcal{N}_{j}(v)$ is defined as the set of vertices $w$ such that there exists a shortest path 
\[
w_0 = v - w_1 - \cdots - w_j = w
\]
with vertex degree parities following the pattern prescribed by $\mathbf{a}$:

\begin{equation}
\label{eq:deg-paritySeq}
\deg(w_s) \equiv a_s \pmod{2} \quad \text{for all } 1 \leq s \leq j.    
\end{equation}
\end{defn}

Thus, $C^{v}_{j, \mathbf{a}}$ consists of vertices at distance $j$ from $v$ that share a common degree parity pattern along at least one shortest path from $v$. 
% These sets induce a partition (or, more generally, a covering) of each layer $\mathcal{N}_{v,j}$ into groups of vertices that are indistinguishable with respect to parity profiles.
There is a natural correspondence between the vertices in $\mathcal{N}_{v,j}$ and the subsets $\mathcal{C}_{j, a}^v$, whereby each vertex is associated with all subsets $\mathcal{C}_{j, a}^v$ to which it belongs. In particular, every vertex belongs to at least one such subset. %As studied in \cite{TBFS}, the unique partition of the vertices into these groups yields information about the reduced DLA, which we are interested in.
% \BT{In order to be more consistent with Bao's presentation of the algorithm, I would phrase it as ""}

As established in Theorems~IV.2 and~IV.8 of \cite{TBFS}, partitioning the vertices into these subsets reveals fundamental structural properties of the reduced dynamical Lie algebra. Specifically, the standard reduced DLA contains the elements:
\begin{equation}
\label{eq:FixedDistX1element}
\begin{aligned}
    \mathcal{X}_{v,j} &:= i \sum_{w \in \mathcal{N}_{j}(v)} X_w, \\
    \mathcal{X}_{\mathcal{C}^v_{j,\mathbf{a}}} &:= i \sum_{w \in \mathcal{C}^v_{j,\mathbf{a}}} X_w.
\end{aligned}
\end{equation}

These containments hold provided the following parity conditions are satisfied:
\begin{enumerate}
    \item For every $j \geq 1$, each vertex $w \in \mathcal{N}_{j}(v)$ has an odd number of neighbors in the preceding layer $\mathcal{N}_{j-1}(v)$:
    \begin{equation}
    \label{eq:ParityAssumption}
    \left| \{ u \in \mathcal{N}_{j-1}(v) \mid (u,w) \in E \} \right| \equiv 1 \pmod{2};
    \end{equation}

    \item For every $j \ge 1$, every $\mathbf{a} \in \mathbb{Z}_2^j$, and every $w \in \mathcal{C}^v_{j,\mathbf{a}}$, $w$ has an odd number of neighbors in $\mathcal{C}^v_{j-1,\mathbf{a}'}$:
    \begin{equation}
    \label{eq:ParitySeqAssumption}
    \left| \{ u \in \mathcal{C}^v_{j-1,\mathbf{a}'} \mid (u,w) \in E \} \right| \equiv 1 \pmod{2},
    \end{equation}
    where $\mathbf{a}' = (a_1,\dots,a_{j-1})$.
\end{enumerate}

Furthermore, the containment results of Eq.~\eqref{eq:FixedDistX1element} extend to any alternative vertex $v' \in V \setminus \{v\}$, provided that $iX_{v'} \in \mathfrak{g}^v_{\Gamma,\mathrm{std}}$ and the parity conditions as in Eq.~\eqref{eq:ParityAssumption}--~\eqref{eq:ParitySeqAssumption} are met with respect to $v'$. When these criteria are satisfied, the elements
\begin{equation}
\label{eq:FixedDistX2element}
\begin{aligned}
    \mathcal{X}_{v',j} &:= i \sum_{w \in \mathcal{N}_{j}(v')} X_{w};\\ 
    \mathcal{X}_{\mathcal{C}^{v'}_{j,\mathbf{a}}} &:= i \sum_{w \in \mathcal{C}^{v'}_{j,\mathbf{a}}} X_w
\end{aligned}
\end{equation}
are also contained within $\mathfrak{g}^v_{\Gamma,\mathrm{std}}$ (see Remark~IV.12 in \cite{TBFS}). 

\begin{rmk}
    \label{rmk:TreesSatisfyParityAssumption}
    Notably, the parity assumptions in Eqs.~\eqref{eq:ParityAssumption}--\eqref{eq:ParitySeqAssumption} are satisfied for all vertices whenever the underlying graph is geodetic. This follows from the definition of a geodetic graph, where any pair of distinct vertices in a geodetic graph is connected by a unique simple path.
\end{rmk}

%\textbf{Distinguished Elements in Reduced DLAs.} The relevance of this partition arises from its connection to the structure of the reduced dynamical Lie algebra. In particular, one can associate to each subset $C^{v}_{j, \mathbf{a}}$ the operator
%\begin{equation}
%    X_{C^{v}_{j, \mathbf{a}}} := i \sum_{w \in C^{v}_{j, \mathbf{a}}} X_w,
%\end{equation}
%where $X_w$ denotes the Pauli-$X$ operator acting on vertex $w$. These operators can be shown to lie in the reduced dynamical Lie algebra generated by the QAOA Hamiltonians. As a consequence, the Lie algebra naturally contains collective operators supported on subsets of vertices sharing the same parity profile (see Theorem IV.8 in~\cite{TBFS}).
A key observation is that the partition of the vertex set $V$ induced by the subsets $\{C^{v}_{j, \mathbf{a}}\}$ directly influences the structure of the Lie algebra. In particular, if these subsets yield a partition of $V \setminus \{v\}$ into a disjoint union of singleton sets, i.e.,
\begin{equation}
    \label{eq:CompleteSplitting}
V \setminus \{v\}=\bigsqcup\limits_{w \in V \setminus \{v\}} \{w\},
\end{equation}
then all single-site operators $iX_w$ are contained in the reduced dynamical Lie algebra. 

The following technical lemma provides the theoretical foundation for Algorithm~\ref{alg:recursive-partition}.

\begin{lem}
    \label{lem:refinement_closure}
    Let $\Gamma = (V, E)$ be a connected geodetic graph, $v \in V$ a distinguished vertex, and $w_1, w_2 \in V$ arbitrary vertices (which may coincide with each other or with $v$). Let $\mathcal{S}_1 \subseteq \mathcal{N}_{w_1, j_1}$ and $\mathcal{S}_2 \subseteq \mathcal{N}_{w_2, j_2}$ be two subsets of $V \setminus \{v\}$ such that the operators
   \begin{equation*}
    iX_{w_1}, \quad iX_{w_2} \quad (\text{for } w_1, w_2 \neq v), \quad \mathcal{X}_{\mathcal{S}_1} := i \sum_{w \in \mathcal{S}_1} X_{w}, \quad \text{and} \quad \mathcal{X}_{\mathcal{S}_2} := i \sum_{w \in \mathcal{S}_2} X_{w}
\end{equation*}
    belong to $\mathfrak{g}^{v}_{\Gamma, \mathrm{std}}$. Then the element
    \begin{equation}
        \label{eq:intersection_element}
        \mathcal{X}_{\mathcal{S}_1 \cap \mathcal{S}_2} := i \sum_{w \in \mathcal{S}_1 \cap \mathcal{S}_2} X_{w} 
    \end{equation}
    supported on the intersection $\mathcal{S}_1 \cap \mathcal{S}_2$ is also contained in $\mathfrak{g}^{v}_{\Gamma, \mathrm{std}}$.
\end{lem}

\begin{proof} 

We first show that for each distance $j \ge 1$ relative to $w_2$, the sliced component
\begin{equation}
    \label{eq:sliced_component}
    \mathcal{X}_{\mathcal{S}_1 \cap \mathcal{N}_{w_2,j}} := i \sum_{w \in \mathcal{S}_1 \cap \mathcal{N}_{w_2,j}} X_{w}
\end{equation}
is contained in $\mathfrak{g}^v_{\Gamma, \mathrm{std}}$. This provides a decomposition of $\mathcal{X}_{\mathcal{S}_1}$ according to vertex distance from $w_2$. Note that if $w_1 = w_2$, all elements in $\mathcal{S}_1$ are equidistant from $w_2$ by definition (lying in $\mathcal{N}_{w_2, j_2}$), making this step immediate.

For distinct centers $w_1 \neq w_2$, we proceed by induction on the distance $t \ge 1$ from $w_2$. Recall that the operators $\mathcal{X}_{w_2, t}$ belong to $\mathfrak{g}^v_{\Gamma, \mathrm{std}}$ (see Eq.~\eqref{eq:FixedDistX2element}).

By Lemma~A.5 in~\cite{TBFS}, the element
\begin{equation*}
    \mathcal{Z}_{\widehat{v},2} := i \sum_{\substack{(w, w') \in E \\ w, w' \neq v}} Z_w Z_{w'}
\end{equation*}
belongs to $\mathfrak{g}^{v}_{\Gamma, \mathrm{std}}$.

For the base case $t=1$, we define the edge operator
\begin{equation*}
    \mathcal{Z}_{\mathcal{N}_{w_2,1}} := \operatorname{ad}^2_{iX_{w_2}} \circ \operatorname{ad}^2_{\mathcal{X}_{w_2,1}}\!\left(\mathcal{Z}_{\widehat{v},2}\right) \in \mathfrak{g}^v_{\Gamma, \mathrm{std}}.
\end{equation*}

Since $\mathcal{Z}_{\mathcal{N}_{w_2,1}}$ is supported strictly on edges incident to $w_2$ (connecting $w_2$ to its $1$-neighborhood $\mathcal{N}_{w_2, 1}$), double-commuting with $\mathcal{Z}_{\mathcal{N}_{w_2,1}}$ isolates the generator terms supported at distance $1$ from $w_2$. Specifically, evaluating $\operatorname{ad}^2_{\mathcal{Z}_{\mathcal{N}_{w_2,1}}}$ on $\mathcal{X}_{\mathcal{S}_1}$ (if $w_2 \notin \mathcal{S}_1$) or on $\mathcal{X}_{\mathcal{S}_1}-iX_{w_2}$ (if $w_2 \in \mathcal{S}_1$) isolates the $1$-neighborhood projection, yielding
\begin{equation*}
    \mathcal{X}_{\mathcal{S}_1 \cap \mathcal{N}_{w_2,1}} \propto 
    \begin{cases}
        \operatorname{ad}^2_{\mathcal{Z}_{\mathcal{N}_{w_2,1}}}\!(\mathcal{X}_{\mathcal{S}_1}), & \text{if } w_2 \notin \mathcal{S}_1, \\[6pt]
        \operatorname{ad}^2_{\mathcal{Z}_{\mathcal{N}_{w_2,1}}}\!\big(\mathcal{X}_{\mathcal{S}_1}-iX_{w_2}\big), & \text{if } w_2 \in \mathcal{S}_1,
    \end{cases}
    \quad \in \mathfrak{g}^v_{\Gamma, \mathrm{std}}.
\end{equation*}

For the inductive step $t > 1$, assume $\mathcal{X}_{\mathcal{S}_1 \cap \mathcal{N}_{w_2, s}} \in \mathfrak{g}^v_{\Gamma, \mathrm{std}}$ for all $s < t$. We construct the auxiliary element
\begin{equation*}
    \mathcal{Z}_{\mathcal{N}_{w_2,t}} := \operatorname{ad}^2_{\mathcal{X}_{w_2,t-1}} \circ \operatorname{ad}^2_{\mathcal{X}_{w_2,t}}\!\left(\mathcal{Z}_{\widehat{v},2}\right) \in \mathfrak{g}^v_{\Gamma, \mathrm{std}},
\end{equation*}
which is supported on edges of $\Gamma$ joining $\mathcal{N}_{w_2, t-1}$ to  $\mathcal{N}_{w_2, t}$. Subtracting the previously established components from $\mathcal{X}_{\mathcal{S}_1}$ yields the element 

\[
\mathcal{X}'_t := \mathcal{X}_{\mathcal{S}_1} - \sum_{s=1}^{t-1} \mathcal{X}_{\mathcal{S}_1 \cap \mathcal{N}_{w_2, s}} \in \mathfrak{g}^v_{\Gamma, \mathrm{std}}.
\]
Acting on $\mathcal{X}'_t$ with $\operatorname{ad}^2_{\mathcal{Z}_{\mathcal{N}_{w_2,t}}}$ annihilates all terms outside of  $\mathcal{N}_{w_2, t}$, establishing that
\begin{equation*}
    \mathcal{X}_{\mathcal{S}_1 \cap \mathcal{N}_{w_2,t}} \propto \operatorname{ad}^2_{\mathcal{Z}_{\mathcal{N}_{w_2,t}}}\!\left(\mathcal{X}'_t\right) \in \mathfrak{g}^v_{\Gamma, \mathrm{std}}.
\end{equation*}
By induction, $\mathcal{X}_{\mathcal{S}_1 \cap \mathcal{N}_{w_2, j}} \in \mathfrak{g}^v_{\Gamma, \mathrm{std}}$ for all $j \ge 1$.

Finally, setting $j = j_2$, we isolate the intersection $\mathcal{S}_1 \cap \mathcal{S}_2 \subseteq \mathcal{N}_{w_2, j_2}$. We define
\begin{equation*}
    \mathcal{A} := \operatorname{ad}^2_{\mathcal{X}_{w_2,j_2-1}} \circ \operatorname{ad}^2_{\mathcal{X}_{\mathcal{S}_1 \cap \mathcal{N}_{w_2,j_2}}}\!\left(\mathcal{Z}_{\widehat{v},2}\right) \in \mathfrak{g}^v_{\Gamma, \mathrm{std}}.
\end{equation*}
Since $\Gamma$ is acyclic, every vertex $w \in \mathcal{N}_{w_2, j_2}$ has a unique parent vertex $p(w)$ in $\mathcal{N}_{w_2, j_2 - 1}$. Thus, $\mathcal{A}$ is supported on edge terms $Z_{p(w)} Z_{w}$ for $w \in \mathcal{S}_1 \cap \mathcal{N}_{w_2, j_2}$. As $\mathcal{S}_2 \subseteq \mathcal{N}_{w_2, j_2}$, double commuting $\mathcal{A}$ with $\mathcal{X}_{\mathcal{S}_2}$ acts non-trivially if and only if a vertex belongs to both $\mathcal{S}_2$ and $\mathcal{S}_1 \cap \mathcal{N}_{w_2, j_2}$. The set inclusion $\mathcal{S}_2 \subseteq \mathcal{N}_{w_2, j_2}$ yields the identity
\begin{equation*}
    (\mathcal{S}_1 \cap \mathcal{N}_{w_2, j_2}) \cap \mathcal{S}_2 = \mathcal{S}_1 \cap \mathcal{S}_2.
\end{equation*}
Because $\Gamma$ is acyclic, each vertex $w \in \mathcal{N}_{w_2, j_2}$ possesses a unique  edge connecting it to $\mathcal{N}_{w_2, j_2 - 1}$. Consequently, every generator $X_w$ for $w \in \mathcal{S}_1 \cap \mathcal{S}_2$ receives a uniform non-zero scaling factor under $\operatorname{ad}^2_{\mathcal{A}}$, establishing that
\begin{equation*}
    \mathcal{X}_{\mathcal{S}_1 \cap \mathcal{S}_2} \propto \operatorname{ad}^2_{\mathcal{A}}\!\left(\mathcal{X}_{\mathcal{S}_2}\right) \in \mathfrak{g}^v_{\Gamma, \mathrm{std}}.
\end{equation*}
\end{proof}
\begin{rmk}
\label{rmk:algorithm_correctness}
Since all vertex subsets intersected throughout the execution of Algorithm~\ref{alg:recursive-partition} satisfy the hypotheses of Lemma~\ref{lem:refinement_closure}, this result establishes the correctness of the algorithm.
\end{rmk}
We conclude this part by rigorously verifying the properties of Algorithm~\ref{alg:recursive-partition}.

\begin{proof}[Proof of Theorem~\ref{thm:algorithm_certificate}]
Statement~1 is an immediate consequence of Lemma~\ref{lem:refinement_closure} when taking the intersection of $\text{\ScissorHollowRight}_{v}$ and $\text{\ScissorHollowRight}_{\text{current}}$, while Statement~2 is a special case of Statement~1. Statement~3 is established from Statement~2 as follows.

It suffices to show that every generator of $\mathfrak{g}^v_{\Gamma,\mathrm{free}}$ belongs to $\mathfrak{g}^v_{\Gamma,\mathrm{std}}$. To this end, consider an edge $(w,w') \in E$ with $w,w' \neq v$. Recalling the definition of $\mathcal{Z}_{\widehat{v}}$ from Eq.~\eqref{eq:reduced_standard}, a direct evaluation of the nested commutator yields
\begin{equation}
\label{eq:RecoverZZ}
-\frac{1}{4} \bigl[iX_{w'}, [iX_{w'}, \mathcal{Z}_{\widehat{v}}]\bigr] = 
\begin{cases}
i \sum\limits_{\substack{u \in V \setminus \{v\} \\ (u,w') \in E}} Z_u Z_{w'} + i Z_{w'}, & \text{if } w' \in \mathcal{N}_v, \\[10pt]
i \sum\limits_{\substack{u \in V \setminus \{v\} \\ (u,w') \in E}} Z_u Z_{w'}, & \text{otherwise.}
\end{cases}
\end{equation}
The two-body terms can be isolated by taking nested Lie brackets with $iX_{w'}$ and $iX_w$. Specifically, for any edge $\{w,w'\} \in E$ with $w, w' \neq v$, double-bracketing sequentially with $iX_{w'}$ and $iX_w$ yields
\begin{equation}
\label{eq:RecoverZZ2}
i Z_w Z_{w'} \propto \bigl[iX_w, \bigl[iX_w, \bigl[iX_{w'}, [iX_{w'}, \mathcal{Z}_{\widehat{v}}]\bigr]\bigr]\bigr] \in \mathfrak{g}^v_{\Gamma,\mathrm{std}}.
\end{equation}
Having established that all such two-body terms belong to $\mathfrak{g}^v_{\Gamma,\mathrm{std}}$, we can isolate the single-body terms. For any vertex $w'$ adjacent to $v$ (i.e., $w' \in \mathcal{N}_v$), Eq.~\eqref{eq:RecoverZZ} yields
\begin{equation}
\label{eq:RecoverZ}
i Z_{w'} = -\frac{1}{4} \bigl[iX_{w'}, [iX_{w'}, \mathcal{Z}_{\widehat{v}}]\bigr] - i \sum_{\substack{u \in V \setminus \{v\} \\ (u,w') \in E}} Z_u Z_{w'} \in \mathfrak{g}^v_{\Gamma,\mathrm{std}},
\end{equation}
where each product term $i Z_u Z_{w'}$ in the sum is already guaranteed to lie in $\mathfrak{g}^v_{\Gamma,\mathrm{std}}$ by Eq.~\eqref{eq:RecoverZZ2}.

Hence, every generator of $\mathfrak{g}^v_{\Gamma,\mathrm{free}}$ is contained in $\mathfrak{g}^v_{\Gamma,\mathrm{std}}$. Since the reverse inclusion $\mathfrak{g}^v_{\Gamma,\mathrm{std}} \subseteq \mathfrak{g}^v_{\Gamma,\mathrm{free}}$ is trivial, the standard reduced algebra coincides with the free one:
\begin{equation}
\label{eq:EqualityOfReducedDLAs}
\mathfrak{g}^v_{\Gamma,\mathrm{std}} = \mathfrak{g}^v_{\Gamma,\mathrm{free}}.
\end{equation}

Finally, since each vertex serves as a root at most once, at most $|V| - 1$ root vertices are processed in total. For a tree graph, both the \texttt{Partition} and \texttt{Refinement} routines run in $\mathcal{O}(|V|)$ time. Consequently, the overall worst-case time complexity is $\mathcal{O}\bigl(|V|\bigl(\abs{V} + \abs{E}\bigr)\bigr)$, establishing Statement~4.
\end{proof}

\section{Invariance and Structure of Reduced Dynamical Lie Algebras}
\label{sec:appendix}

In this appendix, we provide proofs and technical derivations to demonstrate that standard reduced dynamical Lie algebras remain invariant under the action of their corresponding graph automorphism groups.

\begin{prop}
\label{prop:DLAinvariance}
Let $S_n$ denote the symmetric group acting on $n$ qubits by permuting the tensor factors. For any subgroup $G \subseteq S_n$, a permutation $\sigma \in G$ acts on the $n$-qubit Hilbert space $W = (\mathbb{C}^2)^{\otimes n}$ via the unitary operator $U_\sigma$, and induces an action on operators $A$ on $\mathcal{H}$ by conjugation:
\begin{equation}
\sigma(A) = U_\sigma A U_\sigma^\dagger.
\end{equation}

Let $H_1,\dots,H_k$ be Hermitian operators on $W$, and let
\[
\mathfrak{g} := \langle iH_1,\dots,iH_k \rangle_{\mathrm{Lie}}
\]
be the associated real dynamical Lie algebra. Suppose that each generator is $G$-invariant, i.e.,
\[
\sigma(H_j)=H_j \qquad \text{for all } \sigma\in G,\ j=1,\dots,k.
\]
Then $\mathfrak{g}$ is $G$-invariant:
\[
\sigma(A)=A \qquad \text{for all } A\in \mathfrak{g},\ \sigma\in G.
\]
\end{prop}

\begin{proof}
Since $\sigma$ acts by conjugation, it preserves products and hence the Lie bracket:
\[
\sigma([A,B])=[\sigma(A),\sigma(B)].
\]
By assumption, each $iH_j$ is fixed by $G$. It follows inductively that any iterated commutator of the generators is fixed by $G$, and hence every element of $\mathfrak{g}$ is $G$-invariant.
\end{proof}

\medskip

This general observation has the following application in the context of the dynamical Lie algebras that are relevant to our paper.

\begin{lem}
\label{lem:DLAinvariance}
Let $\Gamma=(V,E)$ be a connected graph on $n$ vertices, and let $\mathrm{Aut}(\Gamma)$ denote its automorphism group, acting on $W=(\mathbb{C}^2)^{\otimes n}$ by permuting tensor factors. 
\begin{enumerate}
    \item The dynamical Lie algebra $\mathfrak{g}_{\Gamma,\mathrm{std}}$ is invariant under $\mathrm{Aut}(\Gamma)$.
    \item For any vertex $v\in V$, the standard reduced dynamical Lie algebra $\mathfrak{g}^{v}_{\Gamma,\mathrm{std}}$ is invariant under the stabilizer subgroup
\[
\mathrm{Stab}(v) := \{\sigma \in \mathrm{Aut}(\Gamma) : \sigma(v)=v\}.
\]
\end{enumerate} 
\end{lem}

\begin{proof}
Both generators of $\mathfrak{g}_{\Gamma,\mathrm{std}}$ are invariant under graph automorphisms. Hence they are fixed by $\mathrm{Aut}(\Gamma)$, and the first claim follows from Proposition~\ref{prop:DLAinvariance}.

Similarly, the generators of the reduced algebra $\mathfrak{g}^{v}_{\Gamma,\mathrm{std}}$ are invariant under automorphisms that fix $v$. Thus they are fixed by $\mathrm{Stab}(v)$, and the second claim again follows from Proposition~\ref{prop:DLAinvariance}.
\end{proof}

\begin{cor}
\label{cor:DLAorbit}
Let $\Gamma$ be a connected graph, and suppose that $\mathrm{Aut}(\Gamma)$ is nontrivial. Then the standard dynamical Lie algebra $\mathfrak{g}_{\Gamma,\mathrm{std}}$ cannot contain all single-site operators $iX_w$, $w\in V$.

In particular, the inclusion
\[
\mathfrak{g}_{\Gamma,\mathrm{std}} \subset \mathfrak{g}_{\Gamma,\mathrm{free}}
\]
is strict.

Moreover, for any vertex $v\in V$, if the stabilizer subgroup
\[
\mathrm{Stab}(v) := \{\sigma \in \mathrm{Aut}(\Gamma) : \sigma(v)=v\}
\]
is nontrivial, then the reduced dynamical Lie algebra $\mathfrak{g}^{v}_{\Gamma,\mathrm{std}}$ cannot contain all single-site operators $iX_w$ with $w\neq v$.
\end{cor}
\begin{proof}
By Lemma~\ref{lem:DLAinvariance}, the Lie algebra $\mathfrak{g}_{\Gamma,\mathrm{std}}$ is invariant under $\mathrm{Aut}(\Gamma)$.

Since $\mathrm{Aut}(\Gamma)$ is nontrivial, there exists $\sigma \in \mathrm{Aut}(\Gamma)$ and a vertex $w\in V$ such that $\sigma(w)=w'\neq w$. Then
\[
\sigma(iX_w)=iX_{w'} \neq iX_w,
\]
which contradicts the invariance condition. Hence $iX_w \notin \mathfrak{g}_{\Gamma,\mathrm{std}}$.

The strict inclusion follows since $\mathfrak{g}_{\Gamma,\mathrm{free}}$ contains all single-site Pauli-$X$ terms. The argument for $\mathfrak{g}^{v}_{\Gamma,\mathrm{std}}$ is identical, replacing $\mathrm{Aut}(\Gamma)$ with $\mathrm{Stab}(v)$.

\end{proof}

\end{document}